%% file: fgsaca.tex
\documentclass[runningheads]{llncs}
\usepackage[T1]{fontenc}
\usepackage{graphicx}
\usepackage{amsmath}
\usepackage{ifthen}
\usepackage{xfrac}
\usepackage{tikz}
\usetikzlibrary{external,shapes,decorations.pathreplacing,positioning,calc,math,matrix}
\usepackage{pgfplots}
\pgfplotsset{compat=newest}

\makeatletter
\renewcommand{\@Opargbegintheorem}[4]{%
  #4\trivlist\item[\hskip\labelsep{#3#2\@thmcounterend}]}
\makeatother

\usepackage{xcolor}
\usepackage{caption}
\usepackage{subcaption}

\usepackage{amsfonts}

\let\originalleft\left
\let\originalright\right
\renewcommand{\left}{\mathopen{}\mathclose\bgroup\originalleft}
\renewcommand{\right}{\aftergroup\egroup\originalright}

\newcommand{\reffig}[1]{Fig.\thinspace\ref{#1}}

\newcommand{\lexlt}[0]{<_\mathit{lex}}
\newcommand{\lexgt}[0]{>_\mathit{lex}}

\newcommand{\SA}[0]{\texttt{SA}}
\newcommand{\nss}[0]{\texttt{nss}}
\newcommand{\pss}[0]{\texttt{pss}}
\newcommand{\ang}[1]{\langle#1\rangle}
\newcommand{\abs}[1]{\left|#1\right|}

\newcommand{\set}[1]{\left\{#1\right\}}
\newcommand{\bra}[1]{\left(#1\right)}
\newcommand{\intervCC}[2]{\left[#1 \, ..\,#2\right]}
\newcommand{\intervCO}[2]{\left[#1 \, ..\,#2\right)}
\newcommand{\intervOC}[2]{\left(#1 \, ..\,#2\right]}
\newcommand{\intervOO}[2]{\left(#1 \, ..\,#2\right)}
\newcommand{\G}[0]{\mathcal{G}}
\newcommand{\Ps}[0]{\mathcal{P}}
\newcommand{\Ls}[0]{\mathcal{L}}

\newcommand{\LIBSAIS}[0]{\texttt{libsais}}
\newcommand{\DIVSUFSORT}[0]{\texttt{DivSufSort}}
\newcommand{\FGSACA}[0]{\texttt{FGSACA}}
\newcommand{\DSH}[0]{\texttt{DSH}}
\newcommand{\DS}[0]{\texttt{DS1}}
\newcommand{\GSACA}[0]{\texttt{GSACA}}

\usepackage{algorithm2e}
\SetKwRepeat{Do}{do}{while}
\SetKwFor{RepTimes}{repeat}{times}{end}

\definecolor{c1}{RGB}{255,150,100}
\definecolor{c2}{RGB}{240,150,120}
\definecolor{c3}{RGB}{225,150,140}
\definecolor{c4}{RGB}{210,150,160}
\definecolor{c5}{RGB}{195,150,180}
\definecolor{c6}{RGB}{180,150,200}
\definecolor{c7}{RGB}{165,150,225}
\definecolor{c8}{RGB}{150,150,255}

\begin{document}
\title{On the Optimisation of the GSACA Suffix Array Construction Algorithm}
%
%
\author{Jannik Olbrich\orcidID{0000-0003-3291-7342}
	\and Enno Ohlebusch
	\and Thomas Büchler
	}
\authorrunning{J. Olbrich \and E. Ohlebusch \and T. Büchler}
%
\institute{University of Ulm, 89081 Ulm, Germany\\
\email{\{jannik.olbrich,enno.ohlebusch,thomas.buechler\}@uni-ulm.de}\\
\url{https://www.uni-ulm.de/in/theo}}
\maketitle              
\begin{abstract}
	The suffix array is arguably one of the most important data structures
	in sequence analysis and consequently there is a multitude of suffix sorting
	algorithms.
	However, to this date the \GSACA{} algorithm introduced in 2015
	is the only known non-recursive linear-time
	suffix array construction algorithm (SACA). Despite its
	interesting theoretical properties, there has been little effort in
	improving the algorithm's subpar real-world performance.
	There is a super-linear algorithm \DSH{} which relies on the same sorting
	principle and is faster than \DIVSUFSORT{}, the fastest SACA for
	over a decade.
	This paper is concerned with analysing the sorting principle used in \GSACA{}
	and \DSH{} and exploiting its properties in order to give an optimised
	linear-time algorithm.
	Our algorithm is not only significantly faster than
	\GSACA{} but also outperforms \DIVSUFSORT{} and \DSH{}.


\end{abstract}

\input{intro}

\input{preliminaries}

\section{GSACA}
\label{sec_gsaca}

We start by giving a high level description of the sorting principle based on grouping 
by Baier \cite{baier2015linear,baier2016linear}. Very basically, the suffixes are first
assigned to lexicographically ordered groups, which are then refined until the
suffix array emerges.
The algorithm consists of the following steps.

\paragraph{Initialisation:}
Group the suffixes according to their first character.
\paragraph{Phase I:}
Refine the groups until the elements in each group have the same Lyndon prefix.
\paragraph{Phase II:}
Sort elements within groups lexicographically.

\begin{definition}[Suffix Grouping, adapted from \cite{bertram_et_al_2015}]
	Let $S$ be a string of length $n$ and $\SA$ the corresponding suffix array.
	A \emph{group} $\G$ with \emph{group context} $\alpha$ is a tuple
	$\ang{g_s,g_e,\abs{\alpha}}$ with \emph{group start} $g_s\in\intervCO{0}{n}$
	and \emph{group end} $g_e\in\intervCO{g_s}{n}$ such that the following
	properties hold:
	\begin{enumerate}
		\item All suffixes in $\SA\intervCC{g_s}{g_e}$ share the prefix
			$\alpha$, i.e.\ for all $i\in\SA\intervCC{g_s}{g_e}$ it holds
			$S_i=\alpha S_{i+\abs\alpha}$.
		\item $\alpha$ is a Lyndon word.
	\end{enumerate}
	We say \emph{$i$ is in $\G$} or \emph{$i$ is an element of $\G$} and write $i\in\G$ if and only if $i\in\SA\intervCC{g_s}{g_e}$.
	A \emph{suffix grouping} for $S$ is a set of groups
	$\G_1,\dots,\G_m$, where the groups are pairwise disjoint
	and cover the entire suffix array. Formally, if
	$\G_i = \ang{g_{s,i},g_{e,i},\abs{\alpha_i}}$ for all $i$, then $g_{s,1}
	= 0, g_{e,m} = n-1$ and $g_{s,j} = 1 + g_{e,j-1}$ for all $j\in\intervCC{2}{m}$.
	For $i,j\in\intervCC 1 m$, $\G_i$ is a \emph{lower} \emph{(higher)} group than $\G_j$ if and only if $i<j$ ($i>j$).
	If all elements in a group $\G$ have $\alpha$ as their Lyndon prefix then
	$\G$ is a \emph{Lyndon group}.
	If $\G$ is not a Lyndon group, it is called \emph{preliminary}. Furthermore, a
	suffix grouping is \emph{Lyndon} if all its groups are Lyndon
	groups, and \emph{preliminary} otherwise.
	\label{def_suffix_grouping}
\end{definition}

With these notions, a suffix grouping is created in the initialisation, which is then
refined in Phase I until it is a Lyndon grouping,
and further refined in Phase II until the suffix array emerges.
\reffig{fig_example_grouping} shows a Lyndon grouping with contexts of our running
example.

\begin{figure}
	\newcommand{\brac}[5]{\draw[decorate,decoration={brace,amplitude=5pt}] (#1) -- (#2) node (#5) [midway,yshift=#3] {#4}}

	\centering
		\begin{tikzpicture}

			\foreach \c [count=\i from 0] in {12, 0, 6, 4,10, 1, 7, 3, 2,5,8,9,11} {
				\node[rectangle,rounded corners,minimum size=15pt] (a\i) at (0.75*\i,-2.5){};
				\node (b\i) at (a\i){\strut{\c}};
			}
			\brac{a0.north west}{a0.north east}{+10pt}{\strut\color{c8}\texttt{\$}}{g1};
			\brac{a1.south east}{a1.south west}{-10pt}{\strut\color{c7}\texttt{acedcebceece}}{g2};
			\brac{a2.north west}{a2.north east}{+10pt}{\strut\color{c6}\texttt{bceece}}{g3};
			\brac{a3.north west}{a4.north east}{+10pt}{\strut\color{c5}\texttt{ce}}{g4};
			\brac{a5.north west}{a5.north east}{+10pt}{\strut\color{c4}\texttt{ced}}{g5};
			\brac{a6.north west}{a6.north east}{+10pt}{\strut\color{c3}\texttt{cee}}{g6};
			\brac{a7.north west}{a7.north east}{+10pt}{\strut\color{c2}\texttt{d}}{g7};
			\brac{a8.north west}{a12.north east}{+10pt}{\strut\color{c1}\texttt{e}}{g8};

			\node (g1l) [above=-2pt of g1] {$\mathcal{G}_1$};
			\path let \p1 = (g1l) in let \p2 = (g2) in node (g2l) at (\x2,\y1) {$\mathcal{G}_2$};
			\node (g3l) [above=-2pt of g3] {$\mathcal{G}_3$};
			\node (g4l) [above=-2pt of g4] {$\mathcal{G}_4$};
			\node (g5l) [above=-2pt of g5] {$\mathcal{G}_5$};
			\node (g6l) [above=-2pt of g6] {$\mathcal{G}_6$};
			\node (g7l) [above=-2pt of g7] {$\mathcal{G}_7$};
			\node (g8l) [above=-2pt of g8] {$\mathcal{G}_8$};
		\end{tikzpicture}
	\caption{
		A Lyndon grouping of $\runningexample{}$ with group
		contexts.
	}
	\label{fig_example_grouping}
\end{figure}

In Subsections \ref{subsec_phase_II} and \ref{subsec_phase_I} we explain Phases
II and I, respectively, of our suffix array construction algorithm. Phase II is
described first because it is much simpler.
Subsection \ref{subsec_init} describes how the data structures needed for
Phase I are set up.

\input{phase2}
\input{phase1}

\subsection{Initialisation}
\label{subsec_init}
In the initialisation, the $\pss$-array with markings must be
computed. We use an implementation by Bille et al.\ \cite{bille2019space} and
then add the markings in a right-to-left scan over $\pss$ (for $i=n-1\to0$, if
$\pss[\pss[i]]$ is unmarked, mark $\pss[i]$ and $\pss[\pss[i]]$, otherwise
unmark $\pss[i]$). Further, the initial suffix grouping needs to be
constructed.
For each $c\in\Sigma$ let the \emph{leaf-$c$-bucket} be the interval in $\SA$
containing those $i$ with $S[i]=c$ and $S_i \lexgt S_{i+1}$ and let the
\emph{inner-$c$-bucket} be the interval in $\SA$ containing those $i$ with $S[i] =
c$ and $S_i\lexlt S_{i+1}$.
We compute the leaf and inner bucket sizes in a right-to-left scan over
$S$ (determining for some $i$ whether $S_i\lexlt S_{i+1}$ can be done in constant time
during a right-to-left scan \cite{ko2003space,nong2009linear}).
Computing the prefix sums over the bucket sizes yields the bucket boundaries.
In a second right-to-left scan over $S$
the references (in the array $I$) to the group starts (i.e.\ bucket starts) are set and
the leaves are inserted into the respective leaf-buckets.

\input{eval}
\input{discussion}

\subsubsection{Acknowledgements}
This work was supported by the Deutsche Forschungsgemeinschaft (DFG – German Research Foundation)
(OH 53/7-1).
%
%
%
\bibliographystyle{splncs04}
\bibliography{bib}

\appendix

\end{document}

%% file: intro.tex
\section{Introduction}

The \emph{suffix array} contains the indices of all suffixes of a string
arranged in lexicographical order. It is arguably one of the most important data
structures in \emph{stringology}, the topic of algorithms on strings and
sequences.
It was introduced in 1990 by Manber and Myers for on-line string
searches \cite{sa_intro} and has since been adopted in a wide area of
applications including text indexing and compression
\cite{bioinformaticsAlgorithms}.
Although the suffix array is conceptually very simple, constructing it
efficiently is not a trivial task.

When $n$ is the length of the input text, the suffix array can be constructed
in $\mathcal{O}(n)$ time and $\mathcal{O}(1)$ additional words of working
space when the alphabet is linearly-sortable
(i.e.\ the symbols in the string can be sorted in $\mathcal{O}(n)$ time)
	\cite{goto2017optimal,li2018optimal,nong2013practical}.
However, algorithms with these bounds are not always the fastest in practice.
For instance, \DIVSUFSORT{} has been the fastest SACA for over
a decade although having super-linear worst-case time complexity
\cite{bertram_et_al_2015,fischer2017dismantling}.
	To the best of our knowledge, the currently
	fastest suffix sorter is \LIBSAIS{}, which appeared as source code in 
	February 2021 on Github\footnote{
		\url{https://github.com/IlyaGrebnov/libsais}, last accessed: August 22, 2022
	} and has not been subject to peer review in
	any academic context.
	The author claims that \LIBSAIS{} is an improved implementation of the SA-IS
	algorithm and hence has linear time complexity \cite{nong2009linear}.

The only non-recursive linear-time suffix sorting algorithm \GSACA{} was
introduced in 2015 by Baier and is not competitive,
neither in terms of speed nor in the amount of memory consumed
\cite{baier2015linear,baier2016linear}.
Despite the new algorithm's entirely novel approach and interesting theoretical
properties \cite{franek2017linear}, there has been little effort in optimising it.
In 2021, Bertram et al.\ \cite{bertram_et_al_2015} provided a much faster SACA \DSH{} using the
same sorting principle as \GSACA{}. Their algorithm beats \DIVSUFSORT{}
in terms of speed, but also has super-linear time complexity.

\paragraph{Our Contributions}
We provide a linear-time SACA that relies on the same \emph{grouping} principle
that is employed by \DSH{} and \GSACA{}, but is faster than both.
This is done by exploiting certain properties of Lyndon words that are not
used in the other algorithms.
As a result, our algorithm is more than $11\%$ faster than \DSH{} on
real-world texts and at least $46\%$ faster than Baier's \GSACA{}
implementation.
	Although our algorithm is not on par with \LIBSAIS{} on real-world data, 
	it significantly improves Baier's sorting principle and
	positively answers the question whether the precomputed Lyndon array can be
	used to accelerate \GSACA{} (posed in \cite{bille2019space}).

The rest of this paper is structured as follows:
Section \ref{sec_preliminaries} introduces the definitions and notations used
throughout this paper.
In Section \ref{sec_gsaca}, the grouping principle is investigated and
a description of our algorithm is provided.
In Section \ref{sec_experiments} our algorithm is evaluated experimentally and
compared to other relevant SACAs.
Finally, Section \ref{sec_discussion} concludes this paper and provides an
outlook on possible future research.

%% file: preliminaries.tex
\section{Preliminaries}
\label{sec_preliminaries}

For $i,j\in\mathbb{N}_0$ we denote the set $\set{k\in\mathbb{N}_0 : i \leq k \leq j}$ by the
interval notations $\intervCC{i}{j} = \intervCO{i}{j+1} = \intervOC{i-1}{j} = \intervOO{i-1}{j+1}$.
For an array $A$ we analogously denote the \emph{subarray} from $i$ to $j$ by
$A\intervCC{i}{j} = A\intervCO{i}{j+1} = A\intervOC{i-1}{j} = A\intervOO{i-1}{j+1} = A[i]A[i+1]\dots A[j]$.
We use zero-based indexing, i.e.\ the first entry of the array $A$ is $A[0]$.
A \emph{string} $S$ of \emph{length} $n$ over an \emph{alphabet} $\Sigma$ is a
sequence of $n$ characters from $\Sigma$. We denote the length $n$ of $S$ by
$\abs{S}$ and the $i$'th symbol of $S$ by $S[i-1]$, i.e.\ strings are zero-indexed.
Analogous to arrays we denote the \emph{substring} from $i$ to $j$ by
$S\intervCC{i}{j} = S\intervCO{i}{j+1} =
S\intervOC{i-1}{j} = S\intervOO{i-1}{j+1} = S[i]S[i+1]\dots S[j]$.
For $j > i$ we let $S\intervCC{i}{j}$ be the \emph{empty string} $\varepsilon$.
The \emph{suffix} $i$ of a string $S$ of length $n$ is the substring $S\intervCO{i}{n}$ and is denoted by $S_i$.
Similarly, the substring $S\intervCC{0}{i}$ is a \emph{prefix} of $S$. A suffix (prefix) is \emph{proper} if
$i>0$ ($i + 1 < n$).
	For two strings $u$ and $v$ and an integer $k\geq 0$ we let $uv$ be the
	concatenation of $u$ and $v$ and denote the $k$-times concatenation of $u$
	by $u^k$.

	We assume totally ordered alphabets.
	This induces a total order on strings.
Specifically, we say a string $S$ of length $n$ is \emph{lexicographically smaller} than another string $S'$ of length $m$ if and only if
there is some $\ell\leq\min\set{n,m}$ such that $S\intervCO{0}{\ell} = S'\intervCO{0}{\ell}$ and either $n=\ell<m$ or $S[\ell] < S'[\ell]$.
If $S$ is lexicographically smaller than $S'$ we write $S\lexlt S'$.

A non-empty string $S$ is a \emph{Lyndon word} if and only if $S$ is lexicographically smaller than
all its proper suffixes \cite{duval1983}.
The \emph{Lyndon prefix} of $S$ is the longest prefix of $S$ that is a Lyndon word. We let $\Ls_i$ denote the
Lyndon prefix of $S_i$. Note that a string of length one is always a Lyndon
word, hence the Lyndon prefix of a non-empty string is also non-empty.

	In the remainder of this paper, we assume an arbitrary but fixed string $S$ of
	length $n>1$ over a totally ordered alphabet $\Sigma$ with
	$\abs\Sigma\in\mathcal{O}(n)$.
Furthermore, we assume w.l.o.g.\ that $S$ is \emph{null-terminated}, that is $S[n-1] = \$$ and $S[i] > \$$ for all
$i\in\intervCO{0}{n-1}$.

The \emph{suffix array} $\SA$ of $S$ is an array of length $n$ that contains the indices of the suffixes of $S$
in increasing lexicographical order. That is, $\SA$ forms a permutation of $\intervCO{0}{n}$ and
$S_{\SA[0]} \lexlt S_{\SA[1]} \lexlt \dots \lexlt S_{\SA[n-1]}$.

We assume the RAM model of computation, that is, basic arithmetic operations can be performed in $\mathcal{O}(1)$ time on words of length $\mathcal{O}(\log n)$, where $n$ is the size of the input.
Reading and writing an entry $A[i]$ of an array $A$ can also be performed in constant time when $i$ and $A[i]$ have length in $\mathcal{O}(\log n)$.

\begin{definition}[\pss-tree \cite{bille2019space}]
	Let $\pss$ be the array such that $\pss[i]$
	is the index of the previous smaller suffix for each $i\in\intervCO{0}{n}$ (or -1 if none exists).
	Formally, $\pss[i] := \max{\bra{\set{j\in\intervCO{0}{i} : S_j \lexlt S_i}\cup\set{-1}}}$.
	Note that $\pss$ forms a tree with -1 as the root, in which each $i\in\intervCO{-1}{n}$
	is represented by a node and $\pss[i]$ is the parent of node $i$. We call
	this tree the \emph{\pss-tree}.
	Further, we impose an order on the nodes that corresponds to the order of
	the indices represented by the nodes.
	In particular, if $c_1<c_2<\dots<c_k$ are the children of $i$ (i.e. $\pss[c_1]=\dots=\pss[c_k]=i$), we say $c_k$ is the \emph{last child of $i$}.
	\label{def_pss_tree}
\end{definition}

\newcommand{\runningexample}[0]{\texttt{acedcebceece\$}}

\input{preliminary_example}

Analogous to $\pss[i]$, we define $\nss[i] := \min{\set{j\in\intervOC{i}{n}
: S_j\lexlt S_i}}$ as the next smaller suffix of $i$. Note that
$S_n=\varepsilon$ is smaller than any non-empty suffix of $S$, hence $\nss$ is
well-defined.

In the rest of this paper, we use $S=\runningexample{}$ as our running example. 
\reffig{fig_Lyndon_pss} shows its Lyndon prefixes and the corresponding $\pss$-tree.
\begin{definition}
	Let $\Ps_i$ be the set of suffixes with $i$ as next smaller suffix, that is
	\[
		\Ps_i = \set{j\in\intervCO0i : \nss[j] = i}
	\]
\end{definition}
For instance, in our running example we have $\Ps_{4} = \set{1,3}$ because
$\nss[1] = \nss[3] = 4$.

%% file: preliminary_example.tex
\begin{figure}[t]
	\centering
	\begin{tikzpicture}[
				cyc/.style={circle,draw=black,thick,minimum size=12pt},
			]
		\draw[fill=c7] (-0.3,-0.3) rectangle ++(0.6+0.75*11,0.6);

		\draw[fill=c6] (-0.225+0.75*6,-0.225) rectangle ++(0.45+0.75*5,0.45);
		
		\draw[fill=c5] (-0.225+0.75*4,-0.225) rectangle ++(0.45+0.75*1,0.45);
		\draw[fill=c5] (-0.175+0.75*10,-0.175) rectangle ++(0.35+0.75*1,0.35);
		
		\draw[fill=c4] (-0.225+0.75*1,-0.225) rectangle ++(0.45+0.75*2,0.45);
		
		\draw[fill=c3] (-0.175+0.75*7,-0.175) rectangle ++(0.35+0.75*2,0.35);

		\foreach \c/\col [count=\i from 0] in {a/,c/,e/c1,d/c2,c/,e/c1,b,c,e/c1,e/c1,c,e/c1,\$/c8} {
			\ifthenelse{\i=0 \OR \i=1 \OR \i=4 \OR \i=6 \OR \i=7 \OR \i=10}{ }{
				\draw [fill=\col] (-0.125+0.75*\i,-0.125) rectangle ++(0.25,0.25);
			}
			\node (s\i) at (0.75*\i,0){\strut\texttt{\c}};
			\node (i\i) at (0.75*\i,0.75){\strut\i};
		}

		\foreach \s [count=\i from 0] in {12,0,6,10,4,1,7,3,11,5,2,9,8} {
			\node (sa\i) at (0.75*\i,-0.75){\strut\s};
		}
		\node[anchor=west] (sa) at (-1.2,-0.75) {\strut $\SA$};

		\foreach \s [count=\i from 0] in {12, 4,3,4, 6,6, 12,10,9,10,12,12,13} {
			\node (nss\i) at (0.75*\i,-1.25){\strut\s};
		}
		\node[anchor=west] (nss) at (-1.2,-1.25) {\strut $\nss$};

		\foreach \s [count=\i from 0] in {-1,0,1,1, 0,4, 0,6,7,7, 6,10, -1} {
			\node (pss\i) at (0.75*\i,-1.75){\strut\s};
		}
		\node[anchor=west] (pss) at (-1.2,-1.75) {\strut $\pss$};

		\foreach \y [count=\i from -1] in {-1,0,1,2,2,1,2,1,2,3,3,2,3,0}{
			\node (t\i) at (\i*0.75,-\y*0.5-2.85){\strut \i};
			\node[cyc] (n\i) at (t\i) {};
		}
		\foreach \p [count=\i from 0] in {n-1,n0,n1,n1,n0,n4,n0,n6,n7,n7,n6,n10,n-1}
			\draw[->,thick] (n\i) |- (\p);
	\end{tikzpicture}
	\caption{
		Shown are the Lyndon prefixes of all suffixes of
		$S=\runningexample{}$ and the corresponding suffix array, $\nss$-array,
		$\pss$-array and $\pss$-tree.
		Each box indicates a Lyndon prefix. For instance, the Lyndon prefix of
		$S_7=\texttt{ceece\$}$ is
		$\Ls_7 = \texttt{cee}$. Note that $\Ls_i$ is exactly $S[i]$
		concatenated with the Lyndon prefixes of $i$'s children in the
		$\pss$-tree (see Lemma \ref{lemma_lyndon_childrens}).
		For example, $\Ls_6 = S[6]\Ls_7\Ls_{10} = \texttt{bceece}$.
	}
	\label{fig_Lyndon_pss}
\end{figure}

%% file: phase2.tex
\subsection{Phase II}
\label{subsec_phase_II}

In Phase II we need to refine the Lyndon grouping obtained in Phase I into the
suffix array. Let $\G$ be a Lyndon group with context $\alpha$ and let $i,j\in\G$.
Since $S_i=\alpha S_{i+\abs{\alpha}}$ and $S_j=\alpha S_{j+\abs{\alpha}}$,
we have $S_i \lexlt S_j$ if and only if
$S_{i+\abs{\alpha}} \lexlt S_{j+\abs{\alpha}}$.
Hence, in order to find the lexicographically smallest suffix in $\G$, it
suffices to find the lexicographically smallest suffix $p$ in
$\set{i+\abs{\alpha} : i\in\G}$.
Note that removing $p-\abs\alpha$ from $\G$ and inserting it into a new
group immediately preceding $\G$ yields a valid Lyndon grouping.
We can repeat this process until each element in $\G$ is in its own singleton
group.
As $\G$ is Lyndon, we have $S_{k+\abs{\alpha}}\lexlt S_k$ for each $k\in\G$.
Therefore, if all groups lower than $\G$ are singletons, $p$ can be
determined by a simple scan over $\G$ (by determining which member of
$\set{i+\abs\alpha : i\in\G}$
is in the lowest group).
Consider for instance $\G_4=\ang{3,4,\abs{\texttt{ce}}}$ from \reffig{fig_example_grouping}. We consider $4+\abs{\texttt{ce}} = 6$ and
$10+\abs{\texttt{ce}}=12$. Among them,
$12$ belongs to the lowest group, hence $S_{10}$ is lexicographically smaller than $S_4$.
Thus, we know that $\SA[3] = 10$, remove $10$ from $\G_4$ and repeat the same process with
the emerging group $\G_4'=\ang{4,4,\abs{\texttt{ce}}}$. As $4$ is the only element of $\G_4'$ we
know that $\SA[4] = 4$.

If we refine the groups in lexicographically increasing order (lower to
higher) as just described, each time a group $\G$ is processed, all groups lower
than $\G$ are singletons.
(Note that we assume $S$ to be nullterminated, thus the
lexicographically smallest group is always known to be the singleton containing
$n-1$.)
However, sorting groups in such a way leads to a superlinear time complexity.
Bertram et al.\ \cite{bertram_et_al_2015} provide a fast-in-practice $\mathcal{O}\bra{n\log n}$ algorithm for this,
broadly following the described approach.

In order to get a linear time complexity, we turn this approach on its head like
Baier does \cite{baier2015linear,baier2016linear}:
Instead of repeatedly finding the next smaller suffix in a group, we
consider the suffixes in lexicographically increasing order and for each encountered suffix
$i$, we move all suffixes that have $i$ as the next smaller suffix (i.e.\ those
in $\Ps_i$) to new
singleton groups immediately preceding their respective old groups as described above.
Corollary \ref{cor_different_groups} implies that this procedure is well-defined.
(Intuitively, all suffixes in $\Ps_i$ have different Lyndon prefixes because
those Lyndon prefixes start at different indices but end at the same index $i$,
hence they must be in different Lyndon groups.)
\begin{lemma}
	\label{lemma_Ps_diff}
	For any $j,j'\in\Ps_i$ we have $\Ls_j \neq \Ls_{j'}$ if and only if $j \neq j'$.
\end{lemma}
\begin{proof}
	Let $j,j'\in\Ps_i$ and $j\neq j'$. By definition of $\Ps_i$ we have
	$\nss[j] = \nss[j'] = i$. Since $\Ls_j = S\intervCO{j}{\nss[j]}$ and
	$\Ls_{j'} = S\intervCO{j'}{\nss[j']}$, $\Ls_j$ and $\Ls_{j'}$ have different
	lengths, implying the claim.
\end{proof}
\begin{corollary}
	In a Lyndon grouping, the elements of $\Ps_i$ are in different groups.
	\label{cor_different_groups}
\end{corollary}

Accordingly, Algorithm \ref{alg_phaseII_basic} correctly computes the suffix
array from a Lyndon grouping. A formal proof of correctness is given in \cite{baier2015linear,baier2016linear}.
\reffig{fig_example_phase_II} shows Algorithm \ref{alg_phaseII_basic}
applied to our running example.
\begin{algorithm}[t]
	$A[0] \gets n-1$\;
	\For{$i = 0 \to n-1$}{
		\For{$j\in\Ps_{A[i]}$}{
			Let $k$ be the start of the group containing $j$\;
			remove $j$ from its current group and put it in a new group
			$\ang{k,k,\abs{\Ls_j}}$
			immediately preceding $j$'s old group\;
			$A[k] \gets j$\;
		}
	}
	\caption{Phase II of \GSACA{} \cite{baier2015linear,baier2016linear}}
	\label{alg_phaseII_basic}
\end{algorithm}

\input{phase2_basic_example}

Note that each element $i\in\intervCO{0}{n-1}$ has exactly one next smaller suffix,
hence there is exactly one $j$ with $i\in\Ps_j$ and thus $i$ is inserted exactly once into a new singleton group in Algorithm
\ref{alg_phaseII_basic}.
Therefore, for each group from the Lyndon grouping obtained in Phase I, it suffices to maintain
a single pointer to the current start of this group.
In \cite{baier2015linear}, these pointers are stored at the end of
each group in $A$. This leads to them being scattered in memory, potentially harming
cache performance. Instead, we store them contiguously in a separate array $C$,
which improves cache locality especially when there are few groups.

Besides this minor point, there are two major differences between our Phase II and Baier's, both
are concerned with the iteration over a $\Ps_i$-set.

The first difference is the way in which we determine the elements of $\Ps_i$ for some $i$.
The following observations immediately enable us to iterate over $\Ps_i$.
\begin{lemma}
	\label{lemma_prev_set_pred}
	$\Ps_i$ is empty if and only if $i = 0$ or $S_{i-1} \lexlt S_i$.
	Furthermore, if $\Ps_i\neq\emptyset$ then $i-1\in \Ps_i$.
\end{lemma}
\begin{proof}
	$\Ps_0 = \emptyset$ by definition. Let $i\in\intervCO{1}{n}$. If $S_{i-1}
	\lexgt S_i$ we have $\nss[i-1] = i$ and thus $i-1\in\Ps_i$. Otherwise
	($S_{i-1}\lexlt S_i$), assume there is some $j<i-1$ such that $\nss[j] = i$.
	By definition, $S_j \lexgt S_i$ and $S_j\lexlt S_k$ for each
	$k\in\intervOO{j}{i}$. But by transitivity we also have $S_j \lexgt S_{i-1}$, which is a
	contradiction, hence $\Ps_i$ must be empty.
\end{proof}
\begin{lemma}
	\label{lemma_prev_set_last_child}
	For some $j\in\intervCO{0}{i}$, we have $j\in\Ps_i$ if and only if $j$'s
	last child is in $\Ps_i$, or $j=i-1$ and $S_j\lexgt S_i$.
\end{lemma}
\begin{proof}
	By Lemma \ref{lemma_prev_set_pred} we may assume $\Ps_i\neq\emptyset$ and
	$j+1<i$, otherwise the claim is trivially true.
	If $j$ is a leaf we have $\nss[j]=j+1 < i$ and thus $j\notin\Ps_i$ by definition.
	Hence assume $j$ is not a leaf and has $j'>j$ as last child, i.e.\
	$\pss[j']=j$ and there is no $k>j'$ with $\pss[k]=j$.
	It suffices to show that $j'\in\Ps_i$ if and only if $j\in\Ps_i$.
	Note that $\pss[j']=j$ implies $\nss[j] > j'$.

	$\implies:$
	From $\nss[j'] = i$ and thus $S_k \lexgt S_{j'} \lexgt S_j$ (for
	all $k\in\intervOO{j'}{i}$) we have $\nss[j]\geq i$. Assume
	$\nss[j] > i$. Then $S_i \lexgt S_j$ and thus $\pss[i] =
	j$, which is a contradiction.

	$\impliedby:$ From $S_i \lexlt S_j \lexlt S_{j'}$ we have $\nss[j'] \leq i$. Assume $\nss[j'] < i$
	for a contradiction. For all $k\in\intervOO{j}{j'}$, $\pss[j']=j$ implies
	$S_{k} \lexgt S_{j'}$. Furthermore,
	for all $k\in\intervCO{j'}{\nss[j']}$ we have $S_k\lexgt
	S_{\nss[j']}$ by definition.
	In combination this implies $S_k\lexgt S_{\nss[j']}$ for all $k\in\intervOO{j}{\nss[j']}$.
	As $\nss[j] = i > \nss[j']$ we hence have $\pss[\nss[j']] = j$, which is a contradiction.
\end{proof}
Specifically, (if $\Ps_i$ is not empty) we can iterate over $\Ps_i$ by
walking up the \pss-tree starting from $i-1$ and halting when we encounter a
node that is not the last child of its parent.\footnote{
	Note that $n-1$ is the last child of the artificial root -1. This ensures that we always
	halt before we actually reach the root of the $\pss$-tree.
	Moreover, since the elements in $\Ps_i$ belong to different Lyndon groups (Corollary \ref{cor_different_groups}), the order in which we process them is not important.
}
Baier \cite{baier2015linear,baier2016linear} tests whether $i-1$ ($\pss[j]$) is in $\Ps_i$ by explicitly checking whether $i-1$
($\pss[j]$) has already been written to $A$. This is done by
having an explicit marker for each suffix.
Reading and writing those markers leads to bad cache performance because the
accessed memory locations are hard to predict (for the CPU/compiler).
Lemmata \ref{lemma_prev_set_pred} and \ref{lemma_prev_set_last_child} enable us to avoid
reading and writing those markers. In fact, in our implementation of Phase II, the
array $A$ is the only memory written to that is not always in the cache.
Lemma \ref{lemma_prev_set_pred} tells us whether we need to
follow the $\pss$-chain starting at $i-1$ or not. Namely, this is the case if and only if
$S_{i-1}\lexgt S_i$, i.e.\ $i-1$ is a leaf in the $\pss$-tree.
This information is required when we encounter $i$ in $A$ during the outer for-loop
in Algorithm \ref{alg_phaseII_basic}, thus we \emph{mark} such an entry $i$ in
$A$ if and only if $\Ps_i\neq\emptyset$.
Implementation-wise, we use the most significant bit
(MSB) of an entry to indicate whether it is marked or not.
By definition, we have $S_{i-1}\lexgt S_i$ if and only if $\pss[i] + 1 < i$.
Since $\pss[i]$ must be accessed anyway when $i$ is inserted into $A$ (for
traversing the $\pss$-chain), we can insert $i$ marked or unmarked into $A$.
Further, Lemma \ref{lemma_prev_set_last_child} implies that we must
stop traversing a $\pss$-chain when the current element is not the last child of
its parent. We mark the entries in $\pss$ accordingly, also using the MSB of
each entry.  In the rest of this paper, we assume the $\pss$-array to be marked
in this way.

Consider for instance $i=6$ in our running example. As $6-1=5$ is a leaf (cf.\
\reffig{fig_Lyndon_pss}), we have $5\in\Ps_6$. We can deduce the fact that
$5$ is indeed a leaf from $\pss[6]=0 < 5$ alone. If we had $\pss[6]=5$ instead,
we would have $\nss[5]\neq6$ and thus $5\notin\Ps_6$. Further, $5$ is the last
child of $\pss[5]=4$, so $4\in\Ps_6$. Since $4$ is not the last child of
$\pss[4]=0$, we have $\Ps_6 = \set{4,5}$.

These optimisations unfortunately come at the cost of $2n$
additional bits of working memory for the markings.
However, as they are integrated into $\pss$ and $A$ there are no additional
cache misses.

Let $G[i]$ be the index of the group start pointer of $i$'s group in $C$. Phase
II with our first major improvement compared to Baier's algorithm is shown
in Algorithm \ref{alg_phaseII_concrete}.
\begin{algorithm}[t]
	$A[0] \gets n-1$\;
	\For{$i = 0 \to n-1$}{
		\If(\tcp*[h]{i.e.\ $A[i]$ is marked}){$A[i]-1$ is a leaf in the $\pss$-tree}{
			$p\gets A[i]-1$\;
			\Do{$p$ is the last child of $\pss[p]$}{
				$A[C[G[p]]] \gets p$\;
				$C[G[p]] \gets C[G[p]] + 1$\;
				$p\gets\pss[p]$\;
			}( \tcp*[h]{i.e.\ $\pss[p]$ is marked})
		}
	}
	\caption{
		Concrete implementation of Phase II of \GSACA{}.
		The array $G$ maps each suffix to its Lyndon group and $C$ maps the Lyndon
		groups that resulted from Phase I to their current start.
		The correctness immediately follows from the correctness of Algorithm
		\ref{alg_phaseII_basic} and Lemmata \ref{lemma_prev_set_pred} and
		\ref{lemma_prev_set_last_child}.
	}
	\label{alg_phaseII_concrete}
\end{algorithm}

The second major change concerns the cache-unfriendliness of traversing the $\Ps_i$-sets.
This bad cache performance results from the fact that the next $\pss$-value (and
the group start pointer) cannot be fetched until the current one is in memory.
Instead of traversing the $\Ps_i$-sets one after another, we opt to traversing
multiple such sets in a sort of breadth-first-search manner simultaneously.
Specifically, we maintain a small ($\leq 2^{10}$ elements) queue $Q$ of elements
(nodes in the $\pss$-tree) that can currently be processed.
Then we iterate over $Q$ and process the entries one after another. Parents of last
children are inserted into $Q$ in the same order as the respective children.
After each iteration, we continue to scan over the suffix array and for each
encountered marked entry $i$ insert $i-1$ into $Q$ until we either encounter
an empty entry in $A$ or $Q$ reaches its maximum capacity.
This is repeated until the suffix array emerges.
The queue size could be unlimited, but limiting it ensures that it fits into the
CPU's cache.
\reffig{fig_example_phase_II_bfs} shows our Phase II on the running
example and Algorithm \ref{alg_phaseII_bfs} describes it formally in pseudo code.
Note that this optimisation is only useful when the queue contains many elements, otherwise
there is no time to prefetch the required data for an element in the queue while
others are processed, and we effectively have Algorithm \ref{alg_phaseII_basic}
with some additional overhead. Fortunately, in real world data this is usually the case
and the small overhead for maintaining the queue is more than offset by the
better cache performance (cf.\ Section \ref{sec_experiments}).
\begin{algorithm}[t]
	$A \gets \bra{n-1}\bot^{n-1}$ \tcp*[l]{set $A[0]=n-1$, fill the rest with ``undefined''}
	$Q\gets$ queue containing only $n-1$\;
	$i \gets 1$\tcp*[l]{current index in $A$}
	\While{$Q$ is not empty}{
		$s\gets \mathit{Q.size}()$\;
		\RepTimes(\tcp*[h]{insert elements that are currently in the queue}){$s$}{
			$v\gets \mathit{Q.pop}()$\;
			\If(\tcp*[h]{$v$ is last child of $\pss[v]$}){$\pss[v]$ is marked}{
				$\mathit{Q.push}(\pss[v])$\;
			}
			$A[C[G[v]]] \gets v$\tcp*[l]{insert $v$}
			\lIf(\tcp*[h]{$v-1$ is leaf}){$\pss[v] + 1 < v$}{mark $A[C[G[v]]]$}
			$C[G[v]] \gets C[G[v]] + 1$\tcp*[l]{increment start of $v$'s old group}
		}
		\While(\tcp*[h]{refill the queue}){$\mathit{Q.size}() < w \land i< n \land A[i]\neq\bot$}{
			\If(\tcp*[h]{$A[i]-1$ is leaf}){$A[i]$ is marked}{
				$\mathit{Q.push}(A[i] - 1)$\;
			}
			$i\gets i+1$\;
		}
	}
	\caption{
		Breadth-first approach to Phase II. The constant $w$ is the maximum
		queue size.
	}
	\label{alg_phaseII_bfs}
\end{algorithm}
\input{phase2_bfs_example}

\begin{theorem}
	Algorithm \ref{alg_phaseII_bfs} correctly computes the suffix array from a
	Lyndon grouping.
	\label{lemma_bfs_correct}
\end{theorem}
\begin{proof}
	By Lemmata \ref{lemma_prev_set_pred} and \ref{lemma_prev_set_last_child},
	Algorithms \ref{alg_phaseII_basic} and
	\ref{alg_phaseII_bfs} are equivalent for a maximum queue size of 1.
	Therefore it suffices to show that the result of Algorithm
	\ref{alg_phaseII_bfs} is independent of the queue size.
	Assume for a contradiction that the algorithm inserts two elements $i$ and $j$ with
	$S_i\lexlt S_j$ belonging to the same Lyndon group with context $\alpha$, but in
	a different order as Algorithm \ref{alg_phaseII_basic} would.
	This can only happen if $j$ is inserted earlier than $i$.
	Note that, since $i$ and $j$ have the same Lyndon prefix $\alpha$, the
	\pss-subtrees $T_i$ and $T_j$ rooted at $i$ and $j$, respectively,
	are isomorphic (see \cite{bille2019space}).
	In particular, the path from the rightmost leaf in $T_i$ to $i$ has the same
	length as the path from the rightmost leaf in $T_j$ to $j$.
	Thus, $i$ and $j$ are inserted in the same order as $S_{i+\abs\alpha}$ and
	$S_{j+\abs\alpha}$ occur in the suffix array. Now the claim follows inductively.
\end{proof}


%% file: phase2_basic_example.tex
\begin{figure}[t]
	\centering
	\begin{tikzpicture}
		\definecolor{marked}{RGB}{170,225,170}

		\tikzset{markednode/.style={rectangle,rounded corners,fill=marked,minimum size=15pt}}
		\newcommand{\drawline}[1]{\draw[thick] (0.75*#1-0.375,\H-1em) -- (0.75*\i-0.375,\H+1em)}
		
		\newdimen\Hskip
		\tikzmath{\Hskip=1.2em;}

		\newdimen\H
		\tikzmath{\H=0cm;}

		\foreach \c [count=\i from 0] in {12, 0, 6, 4,10, 1, 7, 3, 2,5,8,9,11} {
			\ifthenelse{\i=0}{
				\node[markednode] at (0.75*\i,\H) {};
			}{}
			\node (a\i) at (0.75*\i,\H){\strut{\c}};
		}
		\foreach \i in {1,2,3,5,6,7,8} { \drawline{\i}; }
		
		\pgfextracty{\H}{\pgfpointanchor{a0}{south}}
		\tikzmath{\H=\H+1em-\Hskip;}
		
		\newdimen\W
		\tikzmath{
			coordinate \C;
			\C=(a12.east)-(a0.west);
			\W=\Cx;
		}
		\tikzset{textnode/.style={anchor=north west,text width=\W,inner sep=0pt}}

		\node[textnode] (t) at (-7.5pt,\H) {
			Since $S$ is nullterminated, $\SA[0] = n-1 = 12$. Hence we insert $\Ps_{12}=\set{0,6,10,11}$.
		};

		\pgfextracty{\H}{\pgfpointanchor{t}{south}}
		\tikzmath{\H=\H-\Hskip;}
		
		\foreach \c [count=\i from 0] in {12, 0, 6, 10,4, 1, 7, 3, 11,2,5,8,9} {
			\ifthenelse{\i<4 \OR \i=8}{
				\node[markednode] at (0.75*\i,\H) {};
			}{}
			\node (a\i) at (0.75*\i,\H){\strut{\c}};
		}
		\foreach \i in {1,2,3,4,5,6,7,8,9} { \drawline{\i}; }
		
		\pgfextracty{\H}{\pgfpointanchor{a0}{south}}
		\tikzmath{\H=\H+1em-\Hskip;}

		\node[textnode] (t) at (-7.5pt,\H) {
			We skip $\SA[1]$ since $\Ps_0=\emptyset$. Thus, $\Ps_6 = \set{4,5}$
			is inserted next.
		};

		\pgfextracty{\H}{\pgfpointanchor{t}{south}}
		\tikzmath{\H=\H-\Hskip;}
		
		\foreach \c [count=\i from 0] in {12, 0, 6, 10,4, 1, 7, 3, 11,5,2,8,9} {
			\ifthenelse{\i<5 \OR \i=8 \OR \i=9}{
				\node[markednode] at (0.75*\i,\H) {};
			}{}
			\node (a\i) at (0.75*\i,\H){\strut{\c}};
		}
		\foreach \i in {1,2,3,4,5,6,7,8,9,10} { \drawline{\i}; }

		\pgfextracty{\H}{\pgfpointanchor{a0}{south}}
		\tikzmath{\H=\H+1em-\Hskip;}

		\node[textnode] (t) at (-7.5pt,\H) {
			Next we have $\Ps_{10} = \set{7,9}$.
		};

		\pgfextracty{\H}{\pgfpointanchor{t}{south}}
		\tikzmath{\H=\H-\Hskip;}
		
		\foreach \c [count=\i from 0] in {12, 0, 6, 10,4, 1, 7, 3, 11,5,9,2,8} {
			\ifthenelse{\i<5 \OR \i=6 \OR \i=8 \OR \i=9 \OR \i=10}{
				\node[markednode] at (0.75*\i,\H) {};
			}{}
			\node (a\i) at (0.75*\i,\H){\strut{\c}};
		}
		\foreach \i in {1,2,3,4,5,6,7,8,9,10,11} { \drawline{\i}; }

		\pgfextracty{\H}{\pgfpointanchor{a0}{south}}
		\tikzmath{\H=\H+1em-\Hskip;}

		\node[textnode] (t) at (-7.5pt,\H) {
			Next we have $\Ps_{4} = \set{1,3}$.
		};

		\pgfextracty{\H}{\pgfpointanchor{t}{south}}
		\tikzmath{\H=\H-\Hskip;}
		
		\foreach \c [count=\i from 0] in {12, 0, 6, 10,4, 1, 7, 3, 11,5,9,2,8} {
			\ifthenelse{\i<11}{
				\node[markednode] at (0.75*\i,\H) {};
			}{}
			\node (a\i) at (0.75*\i,\H){\strut{\c}};
		}
		\foreach \i in {1,2,3,4,5,6,7,8,9,10,11} { \drawline{\i}; }

		\pgfextracty{\H}{\pgfpointanchor{a0}{south}}
		\tikzmath{\H=\H+1em-\Hskip;}

		\node[textnode] (t) at (-7.5pt,\H) {
			The only remaining nonempty $\Ps_i$ are $\Ps_3=\set{2}$ and
			$\Ps_9 = \set{8}$, which are considered in that order. Inserting them
			gives the suffix array.
		};

		\pgfextracty{\H}{\pgfpointanchor{t}{south}}
		\tikzmath{\H=\H-\Hskip;}

		\foreach \c [count=\i from 0] in {12,0,6,10,4,1,7,3,11,5,9,2,8} {
			\node[markednode] at (0.75*\i,\H) {};
			\node (a\i) at (0.75*\i,\H){\strut{\c}};
			\ifthenelse{\i>0}{ \drawline{\i}; }{}
		}
	\end{tikzpicture}
	\caption{
		Refining a Lyndon grouping for $S=\runningexample{}$ (see
		\reffig{fig_example_grouping}) into the suffix array, as done in
		Algorithm \ref{alg_phaseII_basic}.
		Inserted elements are colored green.
	}
	\label{fig_example_phase_II}
\end{figure}

%% file: phase2_bfs_example.tex
\begin{figure}[t]
	\centering
	\begin{tikzpicture}
		\definecolor{marked}{RGB}{170,225,170}
		\definecolor{inserted}{RGB}{150,150,200}

		\tikzset{markednode/.style={rectangle,rounded corners,fill=marked,minimum size=15pt}}
		\tikzset{insertednode/.style={rectangle,rounded corners,fill=inserted,minimum size=15pt}}
		\newcommand{\drawline}[1]{\draw[thick] (0.75*#1-0.375,\H-1em) -- (0.75*\i-0.375,\H+1em)}
		
		\newdimen\Hskip
		\tikzmath{\Hskip=1.2em;}

		\newdimen\H
		\tikzmath{\H=0cm;}
		\newdimen\T
		\tikzmath{\T=-7.5pt-0.75cm*0;}

		\foreach \c [count=\i from 0] in {12, 0, 6, 10,4, 1, 7, 3, 11,2,5,8,9} {
			\ifthenelse{\i<2 \OR \i=8}{
				\node[markednode] at (0.75*\i,\H) {};
			}{
				\ifthenelse{\i < 4}{
					\node[insertednode] at (0.75*\i,\H) {};
				}{}
			}
			\node (a\i) at (0.75*\i,\H){\strut{\c}};
		}
		\foreach \i in {1,2,3,4,5,6,7,8,9} { \drawline{\i}; }

		\newdimen\W
		\tikzmath{
			coordinate \C;
			\C=(a12.east)-(a0.west);
			\W=\Cx + 0.75cm * 1.5 + 1cm;
		}
		\tikzset{textnode/.style={anchor=north west,text width=\W,inner sep=0pt}}
		
		\node[textnode] (t) at (\T,\H+3.2em) {
			The first step is the same as in \reffig{fig_example_phase_II}.
			Note that $\Ps_0=\emptyset$, hence 0 is not marked for further processing.
		};
		
		\pgfextracty{\H}{\pgfpointanchor{a0}{south}}
		\tikzmath{\H=\H+1em-\Hskip;}

		\node[textnode] (t) at (\T,\H) {
			Now $6-1=5$ and $10-1=9$ are inserted into the queue and $6$ and
			$10$ are unmarked.
		};

		\pgfextracty{\H}{\pgfpointanchor{t}{south}}
		\tikzmath{\H=\H-\Hskip;}

		\foreach \c [count=\i from 0] in {12, 0, 6, 10,4, 1, 7, 3, 11,2,5,8,9} {
			\ifthenelse{\i<4 \OR \i=8}{
				\node[markednode] at (0.75*\i,\H) {};
			}{ }
			\node (a\i) at (0.75*\i,\H){\strut{\c}};
		}
		\foreach \i in {1,2,3,4,5,6,7,8,9} { \drawline{\i}; }
		
		\node[anchor=west] (q) at (0.75*12.75,\H) {\strut{$Q=$}};
		\node (q1) at (0.75*13.5+0.5,\H) {\strut{5}};
		\node (q1) at (0.75*13.5+1,\H) {\strut{9}};
		
		\pgfextracty{\H}{\pgfpointanchor{q}{south}}
		\tikzmath{\H=\H+1em-\Hskip;}
		\node[textnode] (t) at (\T,\H) {
			In the next step, the elements in the queue are inserted and replaced in the
			queue with their parents (if they are last children, which happens
			to be the case for 5 and 9).
			Note that they must be inserted in the same order as they appear in $Q$.
		};
		
		\pgfextracty{\H}{\pgfpointanchor{t}{south}}
		\tikzmath{\H=\H-\Hskip;}

		\foreach \c [count=\i from 0] in {12, 0, 6, 10,4, 1, 7, 3, 11,5,9,2,8} {
			\ifthenelse{\i<4 \OR \i=8 \OR \i=9}{
				\node[markednode] at (0.75*\i,\H) {};
			}{
				\ifthenelse{\i=10}{
					\node[insertednode] at (0.75*\i,\H) {};
				}{}
			}
			\node (a\i) at (0.75*\i,\H){\strut{\c}};
		}
		\foreach \i in {1,2,3,4,5,6,7,8,9,10,11} { \drawline{\i}; }
		
		\node[anchor=west] (q) at (0.75*12.75,\H) {\strut{$Q=$}};
		\node (q1) at (0.75*13.5+0.5,\H) {\strut{4}};
		\node (q1) at (0.75*13.5+1.0,\H) {\strut{7}};

		\pgfextracty{\H}{\pgfpointanchor{q}{south}}
		\tikzmath{\H=\H+1em-\Hskip;}
		\node[textnode] (t) at (\T,\H) {
			Neither 4 nor 7 are the last child of their respective parent.
		};
		
		\pgfextracty{\H}{\pgfpointanchor{t}{south}}
		\tikzmath{\H=\H-\Hskip;}

		\foreach \c [count=\i from 0] in {12, 0, 6, 10,4, 1, 7, 3, 11,5,9,2,8} {
			\ifthenelse{\i<4 \OR \i=6 \OR \i=8 \OR \i=9}{
				\node[markednode] at (0.75*\i,\H) {};
			}{
				\ifthenelse{\i=4 \OR \i=10}{
					\node[insertednode] at (0.75*\i,\H) {};
				}{}
			}
			\node (a\i) at (0.75*\i,\H){\strut{\c}};
		}
		\foreach \i in {1,2,3,4,5,6,7,8,9,10,11} { \drawline{\i}; }
		
		\node[anchor=west] (q) at (0.75*12.75,\H) {\strut{$Q$ is empty}};

		\pgfextracty{\H}{\pgfpointanchor{q}{south}}
		\tikzmath{\H=\H+1em-\Hskip;}
		\node[textnode] (t) at (\T,\H) {
			However, we can advance the scan over $A$ and insert $4-1=3$ into
			$Q$.
		};
		
		\pgfextracty{\H}{\pgfpointanchor{t}{south}}
		\tikzmath{\H=\H-\Hskip;}

		\foreach \c [count=\i from 0] in {12, 0, 6, 10,4, 1, 7, 3, 11,5,9,2,8} {
			\ifthenelse{\i<5 \OR \i=6 \OR \i=8 \OR \i=9}{
				\node[markednode] at (0.75*\i,\H) {};
			}{
				\ifthenelse{\i=10}{
					\node[insertednode] at (0.75*\i,\H) {};
				}{}
			}
			\node (a\i) at (0.75*\i,\H){\strut{\c}};
		}
		\foreach \i in {1,2,3,4,5,6,7,8,9,10,11} {
				\draw[thick] (0.75*\i-0.375,\H-1em) -- (0.75*\i-0.375,\H+1em);
		}
		
		\node[anchor=west] (q) at (0.75*12.75,\H) {\strut{$Q=$}};
		\node (q1) at (0.75*13.5+0.5,\H) {\strut{3}};

		\pgfextracty{\H}{\pgfpointanchor{q}{south}}
		\tikzmath{\H=\H+1em-\Hskip;}
		\node[textnode] (t) at (\T,\H) {
			Next, 3 is inserted into $A$. As 3 is the last child of 1, we insert 1 into $Q$ and in the next
			step into $A$. As 1 is not the last child of $\pss[1]=0$, $Q$ is
			now empty.
		};
		
		\pgfextracty{\H}{\pgfpointanchor{t}{south}}
		\tikzmath{\H=\H-\Hskip;}

		\foreach \c [count=\i from 0] in {12, 0, 6, 10,4, 1, 7, 3, 11,5,9,2,8} {
			\ifthenelse{\i<7 \OR \i=8 \OR \i=9}{
				\node[markednode] at (0.75*\i,\H) {};
			}{
				\ifthenelse{\i=7 \OR \i=10}{
					\node[insertednode] at (0.75*\i,\H) {};
				}{}
			}
			\node (a\i) at (0.75*\i,\H){\strut{\c}};
		}
		\foreach \i in {1,2,3,4,5,6,7,8,9,10,11} { \drawline{\i}; }
		
		\node[anchor=west] (q) at (0.75*12.75,\H) {\strut{$Q$ is empty}};

		\pgfextracty{\H}{\pgfpointanchor{q}{south}}
		\tikzmath{\H=\H+1em-\Hskip;}
		\node[textnode] (t) at (\T,\H) {
			We can continue the
			scan over $A$ and insert $3-1=2$ and $9-1=8$ into $Q$.
		};
		
		\pgfextracty{\H}{\pgfpointanchor{t}{south}}
		\tikzmath{\H=\H-\Hskip;}

		\foreach \c [count=\i from 0] in {12, 0, 6, 10,4, 1, 7, 3, 11,5,9,2,8} {
			\ifthenelse{\i<11}{
				\node[markednode] at (0.75*\i,\H) {};
			}{ }
			\node (a\i) at (0.75*\i,\H){\strut{\c}};
		}
		\foreach \i in {1,2,3,4,5,6,7,8,9,10,11} { \drawline{\i}; }
		
		\node[anchor=west] (q) at (0.75*12.75,\H) {\strut{$Q=$}};
		\node (q1) at (0.75*13.5+0.5,\H) {\strut{2}};
		\node (q1) at (0.75*13.5+1.0,\H) {\strut{8}};

		\pgfextracty{\H}{\pgfpointanchor{q}{south}}
		\tikzmath{\H=\H+1em-\Hskip;}
		\node[textnode] (t) at (\T,\H) {
			Finally, the elements in the queue can be inserted and the suffix
			array emerges.
		};

		\pgfextracty{\H}{\pgfpointanchor{t}{south}}
		\tikzmath{\H=\H-\Hskip;}
		
		\foreach \c [count=\i from 0] in {12,0,6,10,4,1,7,3,11,5,9,2,8} {
			\node[markednode] at (0.75*\i,\H) {};
			\node (a\i) at (0.75*\i,\H){\strut{\c}};
		}
		\foreach \i in {1,2,3,4,5,6,7,8,9,10,11,12} { \drawline{\i}; }
	\end{tikzpicture}
	\caption{
		Refining a Lyndon grouping for $S=abccabccbcc\$$ (see \reffig{fig_example_grouping}) into the suffix array using Algorithm \ref{alg_phaseII_bfs}.
		Marked entries are coloured blue while inserted but unmarked elements are
		coloured green. Note that the uncoloured entries are not actually present
		in the array $A$ but only serve to indicate the current Lyndon grouping.
	}
	\label{fig_example_phase_II_bfs}
\end{figure}
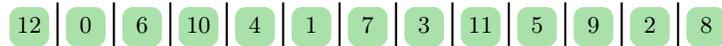

%% file: phase1.tex
\subsection{Phase I}
\label{subsec_phase_I}

In Phase I, a Lyndon grouping is derived from a suffix grouping in which the
group contexts have length (at least) one. That is, the suffixes are
sorted and grouped by their Lyndon prefixes.
Lemma \ref{lemma_lyndon_childrens} describes the relationship between the Lyndon prefixes and the
$\pss$-tree that is essential to Phase I of the grouping principle.
\begin{lemma}
	\label{lemma_lyndon_childrens}
	Let $c_1<\dots<c_k$ be the children of $i\in\intervCO{0}{n}$ in the
	$\pss$-tree as in Definition \ref{def_pss_tree}. Then
	$\Ls_i$ is $S[i]$ concatenated with the Lyndon prefixes of $c_1,\dots,c_k$.
	More formally:
	\begin{align*}
		\Ls_i &= S\intervCO{i}{\nss[i]} \\
			&= S[i]S\intervCO{c_1}{c_2}\dots S\intervCO{c_k}{\nss[i]} \\
			&= S[i]\Ls_{c_1}\dots\Ls_{c_k}
	\end{align*}
\end{lemma}
\begin{proof}
	By definition we have $\Ls_i = S\intervCO{i}{\nss[i]}$.
	Assume $i$ has $k\geq1$ children $c_1<\dots<c_k$ in the $\pss$-tree
	(otherwise $\nss[i]=i+1$ and the claim is trivial).
	For the last child $c_k$ we have $\nss[c_k]=\nss[i]$ from Lemma
	\ref{lemma_prev_set_last_child}. Let $j\in\intervCO{1}{k}$ and assume
	$\nss[c_j] \neq c_{j+1}$. Then we have $\nss[c_j] < c_{j+1}$, otherwise
	$c_{j+1}$ would be a child of $c_j$. As we have $S_{\nss[c_j]}\lexlt
	S_{c_j}$ and $S_{c_j} \lexlt S_{c_{j'}}$ for each $j'\in\intervCO{1}{j}$
	(by induction), we also have $S_{\nss[c_j]} \lexlt S_{i'}$ for each
	$i'\in\intervOO{i}{\nss[c_j]}$. Since $\nss[i] > \nss[c_j]$, $\nss[c_j]$
	must be a child of $i$ in the $\pss$-tree, which is a contradiction.
\end{proof}

We start from the \emph{initial suffix grouping} in which the suffixes are grouped
according to their first characters.
From the relationship between the Lyndon prefixes and the $\pss$-tree in Lemma \ref{lemma_lyndon_childrens} 
one can get the general idea of extending the context
of a node's group with the Lyndon prefixes of its children (in correct order) while
maintaining the sorting \cite{baier2015linear}.
Note that any node is by definition in a higher group than its parent.
Also, by Lemma \ref{lemma_lyndon_childrens} the leaves of the $\pss$-tree
are already in Lyndon groups in the initial suffix grouping.
Therefore, if we consider the groups in lexicographically decreasing order
(i.e.\ higher to lower)
and append the context of the current group to each parent (and insert them into new
groups accordingly), each encountered group is guaranteed to be Lyndon
\cite{baier2015linear}. Consequently, we obtain a Lyndon grouping.
\reffig{fig_example_phase_I} shows this principle applied to our running
example.
\input{phase1_basic_example}

	Formally, the suffix grouping satisfies the following property during Phase I before and after processing a group:
	\begin{property}
		\label{prop_grouping}
		For any $i\in\intervCO{0}{n}$ with children $c_1<\dots<c_k$
		there is $j\in\intervCC{0}{k}$ such that
		\begin{itemize}
			\item $c_1,\dots,c_j$ are in groups that have already been processed,
			\item $c_{j+1},\dots,c_k$ are in groups that have not yet been
				processed, and
			\item the context of the group containing $i$ is
				$S[i]\Ls_{c_1}\dots\Ls_{c_j}$.
		\end{itemize}
		Furthermore, each processed group is Lyndon.
	\end{property}
Additionally and unlike in Baier's original approach, all groups created during our Phase I
are additionally either Lyndon or only contain elements whose Lyndon prefix is different from
the group's context. This has several advantages which are discussed below.
\begin{definition}[Strongly preliminary group]
We call a preliminary group $\G=\ang{g_s,g_e,\abs\alpha}$ \emph{strongly preliminary} if and
only if $\G$ contains only elements whose Lyndon prefix is not
$\alpha$. A preliminary group that is not strongly preliminary is called
\emph{weakly preliminary}.
\end{definition}
\begin{lemma}
	\label{lemma_ls_relation}
	For strings $wu$ and $wv$ over $\Sigma$ with $u\lexlt wu$ and
	$v\lexgt wv$ we have $wu\lexlt wv$.
\end{lemma}
\begin{proof}
	Note that there is no $j\geq1$ such that $wv=w^j$, since otherwise $v$ would be a
	prefix of $wv$ and thus $v\lexlt wv$.
	Hence, there are $k\in\mathbb{N}, \ell\in\intervCO0{\abs{w}}, b\in\Sigma$ and
	$m\in\Sigma^*$ such that $wv = w^kw\intervCO0{\ell}bm$ and $b > w[\ell]$.
	There are two cases:
	\begin{itemize}
		\item There is some $j\geq1$ such that $wu=w^j$.
			\begin{itemize}
				\item If $j\abs w \leq k\abs w + \ell$, then $wu$ is a prefix of
					$wv$.
				\item Otherwise, the first different symbol in $wu$ and $wv$ is
					at index $p = k\abs w+\ell$ and we have
					$(wu)[p] = w^j[p] = w[\ell] < b = (wv)[p]$.
			\end{itemize}
		\item There are $i\in\mathbb{N}, j\in\intervCO0{\abs{w}}, a\in\Sigma$
			and $q\in\Sigma^*$ such that $wu = w^iw\intervCO{0}{j}aq$ and $a < w[j]$.
			\begin{itemize}
				\item If $|w^iw\intervCO{0}{j}| \leq |w^kw\intervCO{0}{\ell}|$, the first
					different symbol is at index $p = |w^iw\intervCO{0}{j}|$ with
					$(wu)[p] = a < w[j] \leq (wv)[p]$.
				\item Otherwise, the first different symbol is at index $p =
					|w^kw\intervCO{0}{\ell}|$ with $(wv)[p] = b > w[\ell] = (wu)[p]$.
			\end{itemize}
	\end{itemize}
	In all cases, the claim follows.
\end{proof}

\begin{lemma}
	For any weakly preliminary group $\G=\ang{g_s,g_e,\abs{\alpha}}$ 
	there is some $g'\in\intervCO{g_s}{g_e}$ such that 
	$\G'=\ang{g_s,g',\abs{\alpha}}$
	is a Lyndon group and
	$\G''=\ang{g'+1,g_e,\abs{\alpha}}$
	is a strongly preliminary group. Splitting $\G$ into $\G'$ and $\G''$
	results in a valid suffix grouping.
	\label{lemma_split_mildly_prel}
\end{lemma}
\begin{proof}
	Let $\G=\ang{g_s,g_e,\abs\alpha}$ be a weakly preliminary group. Let
	$F\subset\G$ be the set of elements from $\G$ whose Lyndon
	prefix is $\alpha$. By Lemma \ref{lemma_ls_relation} we have $S_i \lexlt
	S_j$ for any $i\in F, j\in \G\setminus F$. Hence, splitting $\G$ into two
	groups $\G'=\ang{g_s,g_s+\abs{F}-1,\abs\alpha}$ and $\G''=\ang{g_s+\abs{F},
	g_e, \abs\alpha}$ results in a valid suffix grouping. Note that, by
	construction, the former is a Lyndon group and the latter is strongly preliminary.
\end{proof}
For instance, in \reffig{fig_example_phase_I} there is a group containing
1,4 and 10 with context \texttt{ce}. However, 4 and 10 have this context as Lyndon prefix
while 1 has \texttt{ced}. Consequently, 1 will later be moved to a new group.
Hence, when Baier (and Bertram et al.) create a weakly preliminary group (in
\reffig{fig_example_phase_I} this happens while processing the Lyndon group with context \texttt{e}),
we instead create two groups, the lower containing 4 and 10 and the
higher containing 1.

During Phase I we maintain the suffix grouping using the following data structures:
\begin{itemize}
	\item An array $A$ of length $n$ containing the unprocessed Lyndon groups and
		the sizes of the strongly preliminary groups.
	\item An array $I$ of length $n$ mapping each element $s$ to the start of the
		group containing it. We call $I[s]$ the \emph{group pointer} of $s$.
	\item A list $C$ storing the starts of the already processed Lyndon groups.
\end{itemize}
These data structures are organised as follows.
Let $\G=\ang{g_s,g_e,\abs\alpha}$ be a group.
For each $s\in\G$ we have $I[s] = g_s$.
If $\G$ is Lyndon and has not yet been processed, we also have
$s\in A\intervCC{g_s}{g_e}$ for all $s\in\G$
and $A[g_s]<A[g_s+1] < \dots < A[g_e]$.
If $\G$ is Lyndon and has been processed already, there is some $j$ such that $C[j] = g_s$.
If $\G$ is (strongly) preliminary we have $A[g_s] = g_e + 1 - g_s$
and $A[k] = 0$ for all $k\in\intervOC{g_s}{g_e}$.

In contrast to Baier, we have the Lyndon groups in $A$ sorted and store the sizes of
the strictly preliminary groups in $A$ as well \cite{baier2015linear,baier2016linear}.
The former makes finding the number
of children a parent has in the currently processed group easier and faster.
The latter makes the separate array of length $n$ used by Baier for the
group sizes obsolete \cite{baier2015linear,baier2016linear} and is
made possible by the fact that we only write Lyndon groups to $A$.

As alluded above, we follow Baier's approach and consider the Lyndon groups in lexicographically decreasing
order while updating the groups containing the parents of elements in the current group.

\begin{algorithm}[H]
	$g_e \gets n-1$\;
	\While{$g_e \geq 0$}{
		$g_s\gets I[A[g_e]]$\;
		process group $\ang{g_s,g_e,\bot}$\;
		$g_e\gets g_s-1$\;
	}
	\caption{Phase I: Traversing the groups \cite{baier2015linear,baier2016linear}}
\label{alg_phase1_traversal}
\end{algorithm}
Note that in Algorithm \ref{alg_phase1_traversal}, $g_e$ is always the end of a
Lyndon group. This is due to the fact that a child is by definition
lexicographically greater than its parent. Hence, when a group ends at $g_e$
and all suffixes in $\SA\intervOO{g_e}{n}$ have been processed, the children of
all elements in that group have been processed and it consequently must be
Lyndon. Thus, Algorithm \ref{alg_phase1_traversal} actually results in a Lyndon
grouping.
For a formal proof see \cite{baier2015linear}.

Of course we have to explain how to actually process a Lyndon group. This is
done in the rest of this section.

\newcommand{\As}[0]{\mathcal{A}}

Let $\G = \ang{g_s,g_e,\abs\alpha}$ be the currently processed group
and w.l.o.g.\ assume that no element in $\G$ has the root $-1$ as parent (we do
not have the root in the suffix grouping, thus nodes with the root as parent can
be ignored here).
Furthermore, let $\As$ be the set of parents of elements in $\G$ (i.e.\ $\As =
\set{\pss[i] : i\in\G, \pss[i] \geq 0}$)
and let $\G_1<\dots<\G_k$ be those (necessarily
preliminary) groups containing elements from $\As$.
For each $g\in\intervCC{1}{k}$ let $\alpha_g$ be the context of $\G_g$.

As noted in \reffig{fig_example_phase_I}, we have to consider the number of
children an element in $\As$ has in $\G$. Namely, if a node has multiple
children with the same Lyndon prefix, of course all of them contribute to its
Lyndon prefix. This means that we need to move two parents in $\As$, which are
currently in the same group, to different new groups if they have differing
numbers of children in $\G$.

Let $\As_\ell$ contain those elements from
$\As$ with exactly $\ell$ children in $\G$.
Maintaining Property \ref{prop_grouping} requires that,
after processing $\G$, for some $g\in\intervCC1k$ the elements in
$\As_\ell\cap\G_g$
are in groups with context $\alpha_g\alpha^\ell$.
Note that, for any $\ell<\ell'$, we have
$\alpha_g\alpha^\ell \lexlt \alpha_g\alpha^{\ell'}$.
Consequently, the elements in $\As_\ell\cap\G_g$ must form a lower group than
those in $\As_{\ell'}\cap\G_g$ after $\G$ has been processed
\cite{baier2015linear,baier2016linear}.
To achieve this, first the parents in
$\As_{\abs{\G}}$ are moved to new groups, then those in
$\As_{\abs{\G}-1}$ and so on \cite{baier2015linear,baier2016linear}.

We proceed as follows.
First, determine $\As$ and count how many children each parent has in
$\G$.
Then, sort the parents according to these counts using a bucket sort. Because
the elements of yet unprocessed Lyndon groups must be sorted in $A$, this sort must be stable.
Further, partition each bucket into two sub-buckets, one containing the elements
that should be inserted into Lyndon groups and the other containing those that
Then, for the sub-buckets (in the order of decreasing count; for equal counts: first
strongly preliminary then Lyndon sub-buckets) move the parents into new
groups.\footnote{
	Note that Baier broadly follows the same steps (determine parents, sort them,
	move them to new groups accordingly) \cite{baier2015linear,baier2016linear}.
	However, each individual step is different because of our distinction
	between strongly preliminary, weakly preliminary and Lyndon groups.
}
These steps will now be described in detail.

For brevity, we refer to those elements in $\As$ which have their last
child in $\G$ as \emph{finalists}. Partition $\As_\ell$ into $F_\ell$ and
$N_\ell$, such that the former contains finalists and the latter the
non-finalists.

\newcommand{\key}[1]{\mathit{key}\left(#1\right)}
In order to determine the aforementioned sub-buckets, we
associate a \emph{key} with each element in $\As$
such that (stably) sorting according to these keys yields the desired partitioning.
Specifically, for a fixed $\ell$, let $\key{s} = 2\ell$ for each $s\in F_\ell$ and 
$\key{s} = 2\ell + 1$ for each $s\in N_\ell$.

As we need to sort stably, the bucket sort requires an additional array $B$ of
length $\abs{\G}$, and another array
for the bucket counters.

Finding parents is done using the same $\pss$ array as in Phase II. Since
$A\intervCC{g_s}{g_e}$ is sorted
by increasing index, children of the same parent are in a contiguous part of
$A\intervCC{g_s}{g_e}$. Hence, we determine $\As$ and the keys
within one scan
over $A\intervCC{g_s}{g_e}$. Since in practice most elements have no sibling in
the same Lyndon group, we treat those explicitly.
Specifically, we move $F_1$ to $A\intervCO{g_s}{g_s +
\abs{F_1}}$ and
$N_1$ to $B\intervOC{\abs{\G} - \abs{N_1}}{\abs{\G}}$.
Parents with keys larger than two are written with their keys interspersed to
$B\intervCC{1}{2(\abs{\As} - \abs{\As_1})}$. Interspersing the
keys is done to improve the cache-locality and thus performance.

Then we copy $N_1$ to $A\intervCO{g_s+\abs{F_1}}{g_s+\abs{\As_1}}$. Note
that $\As_1$ is now correctly sorted in $A\intervCO{g_s}{g_s + \abs{\As_1}}$.
Then we sort $\As\setminus\As_1$ by the keys.
That is, we count the frequency of each key, determine the end of each key's
bucket and insert the elements into the correct bucket in
$A\intervCO{g_s+\abs{\As_1}}{g_s+\abs{\As}}$.
\input{mem_layout}
\reffig{fig_mem_layout} shows how the data is organised during the
sorting.

\paragraph{Reordering parents into Lyndon groups}
Let $A'$ be a sub-bucket of length $k$ containing only parents which will now be moved
to a Lyndon group and whose context is extended by $\alpha^q$ for some fixed $q$.
Note that the bucket sort ensures that $A'$ is sorted.
Within each current preliminary group $\G'=\ang{g_s,g_e,\abs\beta}$,
the elements in $A'$ must be
moved to a new Lyndon group following $\G'\setminus A'$.
For each element $s$ in $A'$, we decrement $A[I[s]]$ (i.e. the size of the
group currently containing $s$) and write $s$ to $A[I[s] + A[I[s]]]$.
Afterwards, the new group start must be set (iff $\G'$ is now not
empty) to $I[s] + A[I[s]]$ (the start of the old group plus the remaining size of the old group).
To determine whether $\G'$ is now not empty, we mark inserted
elements in $A$ using the MSB. If $A[I[s]]$ has the MSB set, we do not need to change the group
pointer $I[s]$.

\paragraph{Reordering parents into strongly preliminary groups}
Let $A'$ be a sub-bucket of length $k$ containing only parents which are now moved
to a strongly preliminary group.
The general procedure is similar to the reordering into Lyndon groups, but simpler. First, we decrement
the sizes of the old groups. In a second scan over $A'$, we
set the new group pointer as above, and in a third scan we increment the
sizes of the new groups.

Note that in the reordering step, we iterate two and three times, respectively,
over the elements in a sub-bucket and that in each scan the group pointers are
required.
Furthermore, the group pointers are updated only in the last scan.
As the group pointers are generally scattered in memory, it would be inefficient
to fetch them individually in each scan for two reasons.
Firstly, a single group pointer could occupy an entire cache-line (i.e.\ we mostly keep and transfer
irrelevant data in the cache).
Secondly, the memory accesses are unnecessarily unpredictable.
In order to mitigate these problems,
we pull the group pointers into the temporary array $B$, that is, we set $B[i]
\gets I[A[i + g_s]]$ for each $i\in\intervCO{0}{\abs\As}$.
Of course, in this fetching of group pointers we have the same problems as
before, but during the actual reordering the group pointers can be accessed much
more efficiently.

Note that in contrast to \cite{baier2015linear}, we do not compute the parents
on the fly during
Phase I but instead use the very fast algorithm by Bille et al. \cite{bille2019space}
to compute $\pss$ in advance and then mark the last children. There are two reasons for this, namely
determining the parents on the fly as done in \cite{baier2015linear} requires a kind of
pointer jumping that is
very cache unfriendly and hence slow; and secondly it is not clear how to
efficiently determine on the fly whether a node is the last child of its parent.

Another difference that is speeding up the algorithm is that we only write
Lyndon groups to $A$. This way we do not have to rearrange elements in weakly
preliminary groups when some of their elements are moved to new groups.
Furthermore, it is possible to have the elements in Lyndon groups sorted in $A$
which makes determining the parents and their corresponding keys easier and faster.

%% file: phase1_basic_example.tex
\begin{figure}
	\centering
	\begin{tikzpicture}

		\definecolor{marked}{RGB}{170,225,170}
		\definecolor{done}{RGB}{150,150,150}
		
		\tikzset{markednode/.style={rectangle,rounded corners,fill=marked,minimum size=15pt}}
		\tikzset{donenode/.style={rectangle,rounded corners,fill=done,minimum size=15pt}}
		\newcommand{\drawline}[1]{\draw[thick] (0.75*#1-0.375,\H-1em) -- (0.75*\i-0.375,\H+1em)}
	
		\newcommand{\brac}[4]{\draw[decorate,decoration={brace,amplitude=5pt}] (#1) -- (#2) node (#4) [midway,yshift=+10pt] {#3}}
		
		\newdimen\Hskip
		\tikzmath{\Hskip=1.2em;}

		\newdimen\H
		\tikzmath{\H=0cm;}

		\foreach \c [count=\i from 0] in {12, 0, 6, 1,4,7,10, 3, 2,5,8,9,11} {
			\ifthenelse{\i > 6 \OR \i=0}{
				\node[markednode] (b\i) at (0.75*\i,\H) {};
			}{ }
			\node (a\i) at (0.75*\i,\H){\strut{\c}};
		}
		\foreach \i in {1,2,3,7,8} { \drawline{\i}; }
		\brac{b8.north west}{b12.north east}{\strut\texttt{e}}{context};
		
		\pgfextracty{\H}{\pgfpointanchor{a0}{south}}
		\tikzmath{\H=\H+1em-\Hskip;}

		\newdimen\W
		\tikzmath{
			coordinate \C;
			\C=(a12.east)-(a0.west);
			\W=\Cx;
		}
		\tikzset{textnode/.style={anchor=north west,text width=\W,inner sep=0pt}}
		
		\node[textnode] (t) at (-7.5pt,\H+4.4em) {
			In the initial suffix grouping, the suffixes are grouped according
			to their first characters.
		};

		\node[textnode] (t) at (-7.5pt,\H) {
			The first considered group contains the elements $2,5,8,9$ and
				$11$ and has context \texttt{e}. The parents of the elements are
				$1,4,10$ and $7$, where the former three each have one child in
				the current group and the latter has two. All are in
				the group with context \texttt{c}.
				Thus, we first move 7 to a new group with context
				\texttt{cee} and then 1,4 and 10 to a new group with context
				\texttt{ce}.
		};

		\pgfextracty{\H}{\pgfpointanchor{t}{south}}
		\tikzmath{\H=\H-\Hskip-0.2em;}
		
		\foreach \c [count=\i from 0] in {12, 0, 6, 1,4,10, 7, 3, 2,5,8,9,11} {
			\ifthenelse{\i=0 \OR \i=7 \OR \i=6}{
				\node[markednode] (b\i) at (0.75*\i,\H) {};
			}{
				\ifthenelse{\i>7}{
					\node[donenode] (b\i) at (0.75*\i,\H) {};
				}{ }
			}
			\node (a\i) at (0.75*\i,\H){\strut{\c}};
		}
		\foreach \i in {1,2,3,6,7,8} { \drawline{\i}; }
		\brac{b7.north west}{b7.north east}{\strut\texttt{d}}{context};
		
		\pgfextracty{\H}{\pgfpointanchor{a0}{south}}
		\tikzmath{\H=\H+1em-\Hskip;}

		\node[textnode] (t) at (-7.5pt,\H) {
			Next the group with context \texttt{d} containing 3 is
				processed. The parent of 3 is 1 in a group with context
				\texttt{ce}, so it is moved to a new group with context
				\texttt{ced}. Note that 4 and 10 are now also in a Lyndon group (still
				with context \texttt{ce}).
			};

		\pgfextracty{\H}{\pgfpointanchor{t}{south}}
		\tikzmath{\H=\H-\Hskip-1.2em;}
		
		\foreach \c [count=\i from 0] in {12, 0, 6, 4,10,1, 7, 3, 2,5,8,9,11} {
			\ifthenelse{\i=0 \OR \i=3 \OR \i=4 \OR \i=5 \OR \i=6}{
				\node[markednode] (b\i) at (0.75*\i,\H) {};
			}{
				\ifthenelse{\i>6}{
					\node[donenode] (b\i) at (0.75*\i,\H) {};
				}{ }
			}
			\node (a\i) at (0.75*\i,\H){\strut{\c}};
		}
		\foreach \i in {1,2,3,5,6,7,8} { \drawline{\i}; }
		\brac{b6.north west}{b6.north east}{\strut\texttt{cee}}{context};

		\pgfextracty{\H}{\pgfpointanchor{a0}{south}}
		\tikzmath{\H=\H+1em-\Hskip;}

		\node[textnode] (t) at (-7.5pt,\H) {
			The next processed group contains 7 and has context
			\texttt{cee}. The parent 6 is moved to a new
			group with context \texttt{bcee}. (As 6 is already in a singleton
			group, the actual grouping remains the same except for the context
			of 6's group.)
		};

		\pgfextracty{\H}{\pgfpointanchor{t}{south}}
		\tikzmath{\H=\H-\Hskip-1.2em;}

		\foreach \c [count=\i from 0] in {12, 0, 6, 4,10,1, 7, 3, 2,5,8,9,11} {
			\ifthenelse{\i=0 \OR \i=3 \OR \i=4 \OR \i=5}{
				\node[markednode] (b\i) at (0.75*\i,\H) {};
			}{
				\ifthenelse{\i>5}{
					\node[donenode] (b\i) at (0.75*\i,\H) {};
				}{ }
			}
			\node (a\i) at (0.75*\i,\H){\strut{\c}};
		}
		\foreach \i in {1,2,3,5,6,7,8} {
				\draw[thick] (0.75*\i-0.375,\H-1em) -- (0.75*\i-0.375,\H+1em);
		}
		\brac{b5.north west}{b5.north east}{\strut\texttt{ced}}{context};
		\brac{b3.north west}{b4.north east}{\strut\texttt{ce}}{context};

		\pgfextracty{\H}{\pgfpointanchor{a0}{south}}
		\tikzmath{\H=\H+1em-\Hskip;}

		\node[textnode] (t) at (-7.5pt,\H) {
			The next group again contains only one element, namely 1 with parent
			0. Thus, 0 is put into a new group with context \texttt{aced}.
			Following that, the next group contains 4 and 10, hence their
			parents 0 and 6 are put into new groups with contexts
			\texttt{acedce} and  \texttt{bceece}.
		};

		\pgfextracty{\H}{\pgfpointanchor{t}{south}}
		\tikzmath{\H=\H-\Hskip-1.2em;}
		
		\foreach \c [count=\i from 0] in {12, 0, 6, 4,10,1, 7, 3, 2,5,8,9,11} {
			\ifthenelse{\i=0 \OR \i=2}{
				\node[markednode] (b\i) at (0.75*\i,\H) {};
			}{
				\ifthenelse{\i>2}{
					\node[rectangle, rounded corners, fill=done, minimum size=15pt] at (0.75*\i,\H) {};
				}{}
			}
			\node (a\i) at (0.75*\i,\H){\strut{\c}};
		}
		\foreach \i in {1,2,3,5,6,7,8} { \drawline{\i}; }
		\brac{b2.north west}{b2.north east}{\strut\texttt{bceece}}{context};

		\pgfextracty{\H}{\pgfpointanchor{a0}{south}}
		\tikzmath{\H=\H+1em-\Hskip;}

		\node[textnode] (t) at (-7.5pt,\H) {
			Finally, the only remaining element with a non-root parent is 6
			(with parent 0) in a group with context \texttt{bceece}. Hence,
			0 is put into a Lyndon group with context \texttt{acedcebceece}.
			Afterwards, there is nothing more to do and we obtain the Lyndon grouping from
			\reffig{fig_example_grouping}.
		};
	\end{tikzpicture}
	\caption{
		Refining the initial suffix grouping for $S=abccabccbcc\$$ (see
		\reffig{fig_example_grouping}) into the Lyndon grouping.
		Elements in Lyndon groups are marked gray or green, depending on whether
		they have been processed already. Note that the applied procedure does not entirely
		correspond to our algorithm for Phase I; it only serves to illustrate
		the sorting principle.
	}
	\label{fig_example_phase_I}
\end{figure}

%% file: mem_layout.tex
\begin{figure}[t]
	\centering

	\newcommand{\brac}[4]{\draw[decorate,decoration={brace,amplitude=10pt}] (#1) -- (#2) node[midway,yshift=#3] {#4}}
	\tikzset{rect/.style={rectangle,draw,minimum height=7mm,text depth=0pt}}
	
	\scalebox{0.75}{
	\begin{tikzpicture}
		\node[rect] (Gs) at (0,0) { $F^1$ };
		\tikzmath{coordinate \V;
			\V = (Gs.east) - (Gs.west);
			\remdist = 20em - abs((\Vx)+(\Vy));
		}
		\node[rect,fill=gray,minimum width=\remdist] (Gr) [right=0pt of Gs] {};
		\brac{Gs.north west}{Gr.north east}{+1.75em}{$A\intervCC{g_s}{g_e}$};
		
		\node[rect] (tmpP1) [right=5pt of Gr] {$p_1$};
		\node[rect] (tmpK1) [right=0pt of tmpP1] {$k_1$};
		\node[rect] (tmpD)  [right=0pt of tmpK1] {$\dots$};
		\node[rect] (tmpPl) [right=0pt of tmpD]  {$p_m$};
		\node[rect] (tmpKl) [right=0pt of tmpPl] {$k_m$};
		\tikzmath{coordinate \V;
			\V = (tmpKl.east) - (tmpP1.west);
			\remdist = 20em - 5em - abs((\Vx)+(\Vy));
		}
		\node[rect,fill=gray,minimum width=\remdist] (tmpFill) [right=0pt of tmpKl] {};
		\node[rect,minimum width=5em] (tmpnlc) [right=0pt of tmpFill] {$N^1$};

		\brac{tmpP1.north west}{tmpnlc.north east}{+1.75em}{$B$};


		\node[rect] (Gs) at (0,-1) { $F^1$ };
		\node[rect,minimum width=5em] (Gsl) [right=0pt of Gs] { $N^1$ };
		\tikzmath{coordinate \V;
			\V = (Gsl.east) - (Gs.west);
			\remdist = 20em - abs((\Vx)+(\Vy));
		}
		\node[rect,fill=gray,minimum width=\remdist] (Gr) [right=0pt of Gsl] {};
		
		\node[rect] (tmpP1) [right=5pt of Gr] {$p_1$};
		\node[rect] (tmpK1) [right=0pt of tmpP1] {$k_1$};
		\node[rect] (tmpD)  [right=0pt of tmpK1] {$\dots$};
		\node[rect] (tmpPl) [right=0pt of tmpD]  {$p_m$};
		\node[rect] (tmpKl) [right=0pt of tmpPl] {$k_m$};
		\tikzmath{coordinate \V;
			\V = (tmpKl.east) - (tmpP1.west);
			\remdist = 20em - abs((\Vx)+(\Vy));
		}
		\node[rect,fill=gray,minimum width=\remdist] (tmpFill) [right=0pt of tmpKl] {};



		\node[rect] (Gs) at (0,-2) { $F^1$ };
		\node[rect,minimum width=5em] (Gsl) [right=0pt of Gs] { $N^1$ };
		\node[rect] (Gdots) [right=0pt of Gsl] { $\dots$ };
		\node[rect] (Gk) [right=0pt of Gdots] { $F^{\abs{\G}}$ };
		\node[rect] (Gsk) [right=0pt of Gk] { $N^{\abs{\G}}$ };
		\tikzmath{coordinate \V;
			\V = (Gsk.east) - (Gs.west);
			\remdist = 20em - abs((\Vx)+(\Vy));
		}
		\node[rect,fill=gray,minimum width=\remdist] (Gr) [right=0pt of Gsk] {};
		
		\node[rect,fill=gray,minimum width=20em] (tmpFill) [right=5pt of Gr] {};
	\end{tikzpicture}
	}
	
	\caption{
		Shown is the memory layout during the bucket sort that is applied during
		the processing of a Lyndon group. The data in grey areas is irrelevant.
		$p_1<\dots<p_m$ are the elements in $\mathcal{P}\setminus\mathcal{P}^1$
		and $k_i = \key{p_i}$.
	}
	\label{fig_mem_layout}
\end{figure}
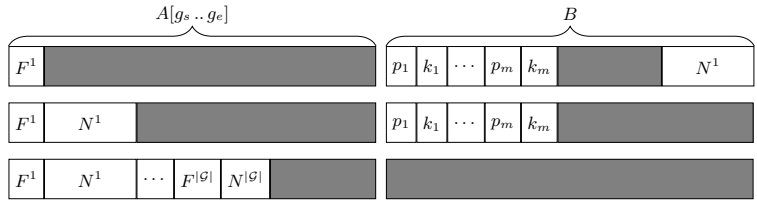

%% file: eval.tex
\section{Experiments}

\label{sec_experiments}
Our implementation \FGSACA{} of the optimised \GSACA{} is publicly
available.\footnote{\url{https://gitlab.com/qwerzuiop/lfgsaca}}

We compare our algorithm with
the \GSACA{} implementation by Baier \cite{baier2015linear,baier2016linear},
and the \texttt{double sort} algorithms \DS{} and
\DSH{} by Bertram et al.\ \cite{bertram_et_al_2015}. The latter two also use the
grouping principle but employ integer sorting and have super-linear time complexity.
\DSH{} differs from \DS{} only in the initialisation: in \DS{} the
suffixes are sorted by their first character while
in \DSH{} up to 8 characters are considered.
	We further include \DIVSUFSORT{} 2.0.2 and \LIBSAIS{} 2.7.1
	since the former is used by Bertram et al.\ as a reference
	\cite{bertram_et_al_2015} and the latter is the currently fastest suffix
	sorter known to us.

All experiments were conducted on a Linux-5.4.0 machine with an AMD EPYC 7742
processor and 256GB of RAM. All SACAs were compiled with GCC 10.3.0
with flags \texttt{-O3 -funroll-loops -march=native -DNDEBUG}.
Each algorithm was executed five times on each test case and we use the mean as
the final result.

We evaluated the algorithms on data from the Pizza \& Chili 
corpus.\footnote{\url{http://pizzachili.dcc.uchile.cl/}}
From the set of real
texts (in the following \texttt{PC-Real}) we included
\texttt{english} (1GiB),
\texttt{dna} (386MiB),
\texttt{sources} (202MiB),
\texttt{proteins} ($1.2$GiB) and
\texttt{dblp.xml} (283MiB).
From the set of real repetitive texts (\texttt{PC-Rep-Real}) we included
\texttt{cere} (440MiB),
\texttt{einstein.en.txt} (446MiB),
\texttt{kernel} (247MiB) and
\texttt{para} (410MiB).
Furthermore, from the artificial repetitive texts (\texttt{PC-Rep-Art}) we
included \linebreak
\texttt{fib41} (256MiB),
\texttt{rs.13} (207MiB) and
\texttt{tm29} (256MiB).

In order to test the algorithms on large inputs for which 32-bit integers
are not sufficient, we also use a dataset containing larger texts
(\texttt{Large}), namely
the first $10^{10}$ bytes of the English Wikipedia dump from
01.06.2022\footnote{
	\url{https://dumps.wikimedia.org/enwiki/20220601/}
}
(9.4GiB) 
and the human DNA concatenated with itself\footnote{
	\url{https://www.ncbi.nlm.nih.gov/assembly/GCF\_000001405.40}
}
(6.3GiB).

For each of the datasets and algorithms, we determined the average
time and memory used, relative to the
input size. The results are shown in \reffig{fig_times}.

All algorithms were faster on the more repetitive datasets, on which the
differences between the algorithms were also smaller.
On all datasets, our algorithm is between $46\%$ and $60\%$ faster than \GSACA{}
and compared to \DSH{} about $2\%$ faster on repetitive data, over $11\%$ faster
on \texttt{PC-Real} and over $13\%$ faster on \texttt{Large}.

Especially notable is the difference in the time required for Phase II: Our Phase
II is between $33\%$ and $50\%$ faster than Phase II of \DSH{}.
Our Phase I is also faster than Phase I of \DS{} by a similar margin.
Conversely, Phase I of \DSH{} is much faster than our Phase I. However,
this is only due to the more elaborate construction of the
initial suffix grouping as demonstrated by the much slower Phase I of \DS{}.
Our initialisation requires constructing and marking the $\pss$-array and is
thus much slower than the initialisation of \GSACA{} and \DS{}. (Note that
in \GSACA{} the unmarked $\pss$-array is also computed, but on the fly during
Phase I.) For \DSH{}, the time
required for the initialisation is much more dependent on the dataset, in
particular it seems to be faster on repetitive data.
	Compared to \FGSACA{}, \LIBSAIS{} is 46\%, 34\%, 31\% and 3\% faster on \texttt{Large},
	\texttt{PC-Real}, \texttt{PC-Rep-Real} and \texttt{PC-Rep-Art},
	respectively.

Memory-wise, for 32-bit words, \FGSACA{} uses about $8.83$ bytes per input character, while \DS{}
and \DSH{} use $8.94$ and $8.05$ bytes/character, respectively.
\GSACA{} always uses 12 bytes/character.
On \texttt{Large}, \FGSACA{} expectedly requires about twice as much memory. For
\DS{} and \DSH{} this is not the case, mostly because they use 40-bit integers for
the additional array of length $n$ that they require (while we use 64-bit
integers).
\DIVSUFSORT{} requires only a small constant amount of working memory
and \LIBSAIS{} never exceeded 21kiB of working memory on our test data.

\input{eval_fig}

%% file: eval_fig.tex
\begin{figure}[t]
\centering

\begin{tikzpicture}
	\definecolor{fntcol}{RGB}{0,0,0}
	\pgfplotsset{
        show sum on top/.style={
            /pgfplots/scatter/@post marker code/.append code={%
                \node[
                    at={(normalized axis cs:%
                            \pgfkeysvalueof{/data point/x},%
                            \pgfkeysvalueof{/data point/y})%
                    },
                    anchor=west,
					rotate=90,
                ]
				{\small\color{fntcol}\pgfmathprintnumber{\pgfkeysvalueof{/data point/y}}};
            },
        },
    }

	\begin{axis}[
		x=14pt,
		y=50pt,
		ybar stacked,
		bar width=11pt,
		ymin=0,
		ymax=2.7,
		nodes near coords={}, 
		ylabel={Time (\sfrac{s}{10 MiB})},
		ytick distance=1,
		symbolic x coords={DivSufSort, FGSACA, DSH, DS1, GSACA},
		xtick=data,
		xticklabels={,,},
		xmajorticks=false,
		at={(0,0)},
		anchor=south west,
		legend image code/.code={
            \draw[#1, draw=none] (0cm,0.0cm) rectangle (1em,1em);
        },
		extra y ticks = 0.51,
		extra y tick labels={},
		extra y tick style={grid=major,major grid style={draw=red}}
	]
        \addplot+[ybar,draw=none] plot coordinates {
                        (DivSufSort, 0)
                        (FGSACA, 0.19293)
                        (DSH, 0.10277)
                        (DS1, 0.027868)
                        (GSACA, 0.077509)
                };
				\label{plot_gsaca_init};
        \addplot+[ybar,draw=none] plot coordinates {
                        (DivSufSort, 0)
                        (FGSACA, 0.14056)
                        (DSH, 0.10814)
                        (DS1, 0.26173)
                        (GSACA, 0.46102)
                };
				\label{plot_gsaca_phase1};
        \addplot+[ybar,draw=none,show sum on top] plot coordinates {
                        (DivSufSort, 0)
                        (FGSACA, 0.19422)
                        (DSH, 0.32976)
                        (DS1, 0.36884)
                        (GSACA, 0.441)
                };
				\label{plot_gsaca_phase2};
        \addplot+[ybar,draw=none,show sum on top] plot coordinates {
                        (DivSufSort, 1.5811)
                        (FGSACA, 0)
                        (DSH, 0)
                        (DS1, 0)
                        (GSACA, 0)
                };
	\end{axis}
	\node at([yshift=0.5em]current axis.north) {\texttt{PC-Rep-Art}};
		
	\tikzmath{
		coordinate \topleft;
		coordinate \last;
		\topleft1=(current axis.south west);
		\last=(current axis.south east);
	}

	\begin{axis}[
		x=14pt,
		y=50pt,
		ybar stacked,
		bar width=11pt,
		ymin=0,
		ymax=2.7,
		yticklabels={,,},
		nodes near coords={}, 
		symbolic x coords={DivSufSort, FGSACA, DSH, DS1, GSACA},
		xtick=data,
		xticklabels={,,},
		xmajorticks=false,
		at={(\lastx+5pt,\lasty)},
		anchor=south west,
		extra y ticks = 0.50,
		extra y tick labels={},
		extra y tick style={grid=major,major grid style={draw=red}}
	]
        \addplot+[ybar,draw=none] plot coordinates {
                        (DivSufSort, 0)
                        (FGSACA, 0.23823)
                        (DSH, 0.18385)
                        (DS1, 0.030751)
                        (GSACA, 0.073301)
                };
        \addplot+[ybar,draw=none] plot coordinates {
                        (DivSufSort, 0)
                        (FGSACA, 0.31497)
                        (DSH, 0.20674)
                        (DS1, 0.50411)
                        (GSACA, 1.1734)
                };
        \addplot+[ybar,draw=none,show sum on top] plot coordinates {
                        (DivSufSort, 0)
                        (FGSACA, 0.17888)
                        (DSH, 0.35557)
                        (DS1, 0.39178)
                        (GSACA, 0.51433)
                };
        \addplot+[ybar,draw=none,show sum on top] plot coordinates {
                        (DivSufSort, 0.8658)
                        (FGSACA, 0)
                        (DSH, 0)
                        (DS1, 0)
                        (GSACA, 0)
                };
	\end{axis}
	\node at([yshift=0.5em]current axis.north) {\texttt{PC-Rep-Real}};
	
	\tikzmath{
		\last=(current axis.south east);
	}

	\begin{axis}[
		x=14pt,
		y=50pt,
		ybar stacked,
		bar width=11pt,
		ymin=0,
		ymax=2.7,
		yticklabels={,,},
		nodes near coords={}, 
		symbolic x coords={DivSufSort, FGSACA, DSH, DS1, GSACA},
		xtick=data,
		xticklabels={,,},
		xmajorticks=false,
		at={(\lastx+5pt,\lasty)},
		anchor=south west,
		extra y ticks = 0.58,
		extra y tick labels={},
		extra y tick style={grid=major,major grid style={draw=red}}
	]
        \addplot+[ybar,draw=none] plot coordinates {
                        (DivSufSort, 0)
                        (FGSACA, 0.22233)
                        (DSH, 0.31804)
                        (DS1, 0.029567)
                        (GSACA, 0.071987)
                };
        \addplot+[ybar,draw=none] plot coordinates {
                        (DivSufSort, 0)
                        (FGSACA, 0.39651)
                        (DSH, 0.28835)
                        (DS1, 0.62424)
                        (GSACA, 1.4719)
                };
        \addplot+[ybar,draw=none,show sum on top] plot coordinates {
                        (DivSufSort, 0)
                        (FGSACA, 0.26378)
                        (DSH, 0.39061)
                        (DS1, 0.44339)
                        (GSACA, 0.73138)
                };
        \addplot+[ybar,draw=none,show sum on top] plot coordinates {
                        (DivSufSort, 0.98472)
                        (FGSACA, 0)
                        (DSH, 0)
                        (DS1, 0)
                        (GSACA, 0)
                };
	\end{axis}
	\node at([yshift=0.5em]current axis.north) {\texttt{PC-Real}};
	
	\tikzmath{
		\last=(current axis.south east);
	}

	\begin{axis}[
		x=14pt,
		y=50pt,
		ybar stacked,
		bar width=11pt,
		ymin=0,
		ymax=2.7,
		yticklabels={,,},
		nodes near coords={}, 
		symbolic x coords={DivSufSort, FGSACA, DSH, DS1},
		xtick=data,
		xticklabels={,,},
		xmajorticks=false,
		at={(\lastx+5pt,\lasty)},
		anchor=south west,
		extra y ticks = 0.91,
		extra y tick labels={},
		extra y tick style={grid=major,major grid style={draw=red}}
	]
        \addplot+[ybar,draw=none] plot coordinates {
                        (DivSufSort, 0)
                        (FGSACA, 0.25844)
                        (DSH, 0.37266)
                        (DS1, 0.032095)
                };
        \addplot+[ybar,draw=none] plot coordinates {
                        (DivSufSort, 0)
                        (FGSACA, 0.71022)
                        (DSH, 0.52412)
                        (DS1, 1.0317)
                };
        \addplot+[ybar,draw=none,show sum on top] plot coordinates {
                        (DivSufSort, 0)
                        (FGSACA, 0.72073)
                        (DSH, 1.0588)
                        (DS1, 1.0414)
                };
        \addplot+[ybar,draw=none,show sum on top] plot coordinates {
                        (DivSufSort, 1.9742)
                        (FGSACA, 0)
                        (DSH, 0)
                        (DS1, 0)
                };
	\end{axis}
	\node at([yshift=0.5em]current axis.north) {\texttt{Large}};

	\tikzmath{
		coordinate \topright;
		\topright1=(current axis.south east);
	}


	\definecolor{fntcol}{RGB}{255,255,255}
	\pgfplotsset{every x tick label/.append style={font=\ttfamily\small}}

	\begin{axis}[
		x=14pt,
		y=4pt,
		ybar stacked,
		bar width=11pt,
		ymin=0,
		ymax=17,
		y dir=reverse,
		ylabel={extra bytes/$n$},
		symbolic x coords={DivSufSort, FGSACA, DSH, DS1, GSACA},
		nodes near coords={},
		xtick=data,
		x tick label style={rotate=90,anchor=east},
		xtick style={draw=none},
		at={(0,-1em)},
		anchor=north west,
	]
		\addplot+[ybar,draw=none,show sum on top,fill=gray] plot coordinates {
			(DivSufSort, 0.00000)
			(FGSACA, 9.84)
			(DSH, 11.37)
			(DS1, 12.63)
			(GSACA, 12)
		};
	\end{axis}
	\tikzmath{
		\last=(current axis.north east);
	}
	\begin{axis}[
		x=14pt,
		y=4pt,
		ybar stacked,
		bar width=11pt,
		ymin=0,
		ymax=17,
		y dir=reverse,
		yticklabels={,,},
		symbolic x coords={DivSufSort, FGSACA, DSH, DS1, GSACA},
		nodes near coords={},
		xtick=data,
		x tick label style={rotate=90,anchor=east},
		xtick style={draw=none},
		at={(\lastx+5pt,\lasty)},
		anchor=north west,
	]
		\addplot+[ybar,draw=none,show sum on top,fill=gray] plot coordinates {
			(DivSufSort, 0.00000)
			(FGSACA, 8.61)
			(DSH, 6.84)
			(DS1, 7.56)
			(GSACA, 12)
		};
	\end{axis}
	\tikzmath{
		\last=(current axis.north east);
	}
	\begin{axis}[
		x=14pt,
		y=4pt,
		ybar stacked,
		bar width=11pt,
		ymin=0,
		ymax=17,
		y dir=reverse,
		yticklabels={,,},
		symbolic x coords={DivSufSort, FGSACA, DSH, DS1, GSACA},
		nodes near coords={},
		xtick=data,
		x tick label style={rotate=90,anchor=east},
		xtick style={draw=none},
		at={(\lastx+5pt,\lasty)},
		anchor=north west,
	]
		\addplot+[ybar,draw=none,show sum on top,fill=gray] plot coordinates {
			(DivSufSort, 0.00000)
			(FGSACA, 8.4)
			(DSH, 7.01)
			(DS1, 7.83)
			(GSACA, 12)
		};
	\end{axis}
	\tikzmath{
		\last=(current axis.north east);
	}
	\begin{axis}[
		x=14pt,
		y=4pt,
		ybar stacked,
		bar width=11pt,
		ymin=0,
		ymax=17,
		y dir=reverse,
		yticklabels={,,},
		symbolic x coords={DivSufSort, FGSACA, DSH, DS1},
		nodes near coords={},
		xtick=data,
		x tick label style={rotate=90,anchor=east},
		xtick style={draw=none},
		at={(\lastx+5pt,\lasty)},
		anchor=north west,
	]
		\addplot+[ybar,draw=none,show sum on top,fill=gray] plot coordinates {
			(DivSufSort, 0.00000)
			(FGSACA, 16.539)
			(DSH, 8.6417)
			(DS1, 9.29616)
		};
	\end{axis}


	\matrix[
		matrix of nodes,
		anchor=north,
		fill=white,
		inner sep=0pt,
		column sep = 0.3em,
		column 1/.style={nodes={anchor=center}},
		column 2/.style={nodes={anchor=west}, font=\strut},
		column 3/.style={nodes={anchor=center}},
		column 4/.style={nodes={anchor=west}, font=\strut},
		column 5/.style={nodes={anchor=center}},
		column 6/.style={nodes={anchor=west}, font=\strut},
	]
	at ([yshift=-14em]$(\topleft1)!0.5!(\topright1)$){
		\phantom{a}\ref{plot_gsaca_init}& Initialisation\phantom{a} &
		\ref{plot_gsaca_phase1}& Phase I\phantom{a} &
		\ref{plot_gsaca_phase2}& Phase II\phantom{a} \\};
\end{tikzpicture}

\caption{
	Normalised running time and working memory averaged for each category.
	The horizontal red line indicates the time for \LIBSAIS{}.
	For \texttt{Large} we did not test \GSACA{} because Baier's reference
	implementation only supports 32-bit words.
}
\label{fig_times}
\end{figure}
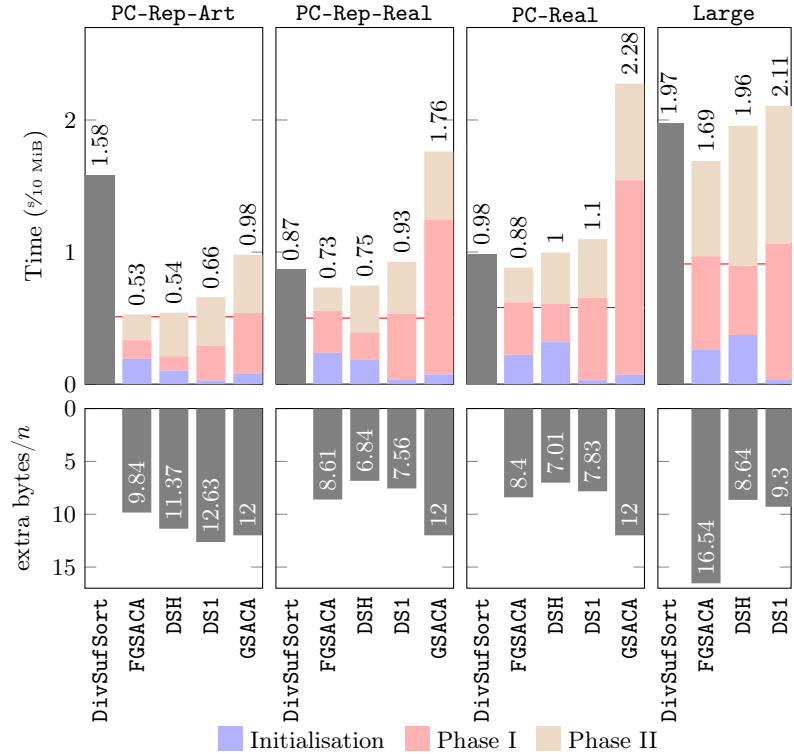

%% file: discussion.tex
\section{Concluding Remarks}
\label{sec_discussion}

We only considered single threaded suffix array construction. As modern
computers gain their processing power more and more through parallelism, it may
be worthwhile to spend effort on trying to parallelise our algorithm. For
instance, while processing a final group, all steps besides the reordering of
the parents are entirely independent of other groups.

Implementation-wise, in the case that 32-bit integers are not sufficient it
may be worthwhile to use 40-bit integers instead of 64-bit integers for our
auxiliary data structures.

Also, the results for \DSH{} and \DS{} indicate that it may be useful to
use an already refined suffix grouping as a starting point for our Phase I, as
this enables us to skip many refinement steps.
